\documentclass[onecolumn,aps,nofootinbib,superscriptaddress,tightenlines,notitlepage,11pt,pra]{revtex4-1}
\usepackage{newfloat}
\usepackage{latexsym}
\usepackage{dcolumn}
\usepackage{graphicx}
\usepackage{amssymb}
\usepackage{amsmath}
\usepackage{amsfonts}
\usepackage{mathrsfs}
\usepackage{float}
\usepackage{color}
\usepackage{algorithm}
\usepackage{algpseudocode}
\usepackage{xcolor}
\usepackage{courier}
\usepackage{enumitem}
\usepackage{bbold}
\usepackage[ breaklinks,
       colorlinks = true,
             linkcolor = blue,
             urlcolor  = blue,
             citecolor = red,
             anchorcolor = blue,
]{hyperref}

\newcommand{\norm}[1]{\lVert #1\rVert}
\newcommand{\Tr}[1]{\mathrm{Tr}\left[#1\right]}

\newcommand{\bra}[1]{\langle #1|}
\newcommand{\expB}[1]{\exp\left(#1 \right)}
\newcommand{\ket}[1]{|#1\rangle}
\newcommand{\logs}{\log}

\def\cD{\mathcal D}
\def\cE{\mathcal{E}}
\def\HE{\mathcal{F}}

\def\X{X}
\def\Y{B}

\DeclareMathOperator*{\argmax}{arg\,max}

\algdef{SE}[DOWHILE]{Do}{doWhile}{\algorithmicdo}[1]{\algorithmicwhile\ #1}%

\newtheorem{theorem}{Theorem}[section]

\newtheorem{prop}[theorem]{Proposition}
\newtheorem{lemma}[theorem]{Lemma}

\newenvironment{proof}[1][Proof]{\noindent\textbf{#1.} }{\ \rule{0.5em}{0.5em}}
\newtheorem{definition}[theorem]{Definition}
\newtheorem{rmk}[theorem]{Remark}

\allowdisplaybreaks

\begin{document}

\title{Computing Quantum Channel Capacities}

\author{Navneeth Ramakrishnan}
\altaffiliation{These authors contributed equally.}
\affiliation{Department of Computing, Imperial College London, United Kingdom}

\author{Raban Iten}
\altaffiliation{These authors contributed equally.}
\affiliation{Institute for Theoretical Physics, ETH Z\"urich, Switzerland}

\author{Volkher B.~Scholz}
\affiliation{Department of Physics and Astronomy, Ghent University, Belgium}

\author{Mario Berta}
\affiliation{Department of Computing, Imperial College London, United Kingdom}


\begin{abstract}
The capacity of noisy quantum channels characterizes the highest rate at which information can be reliably transmitted and it is therefore of practical as well as fundamental importance. Capacities of classical channels are computed using alternating optimization schemes, called Blahut-Arimoto algorithms. In this work, we generalize classical Blahut-Arimoto algorithms to the quantum setting. In particular, we give efficient iterative schemes to compute the capacity of channels with classical input and quantum output, the quantum capacity of less noisy channels, the thermodynamic capacity of quantum channels, as well as the entanglement-assisted capacity of quantum channels. We give rigorous \textit{a priori} and \textit{a posteriori} bounds on the estimation error by employing quantum entropy inequalities and demonstrate fast convergence of our algorithms in numerical experiments.
\end{abstract}


\maketitle

\begin{figure}[t]
    \center
    \includegraphics[width=1.\textwidth]{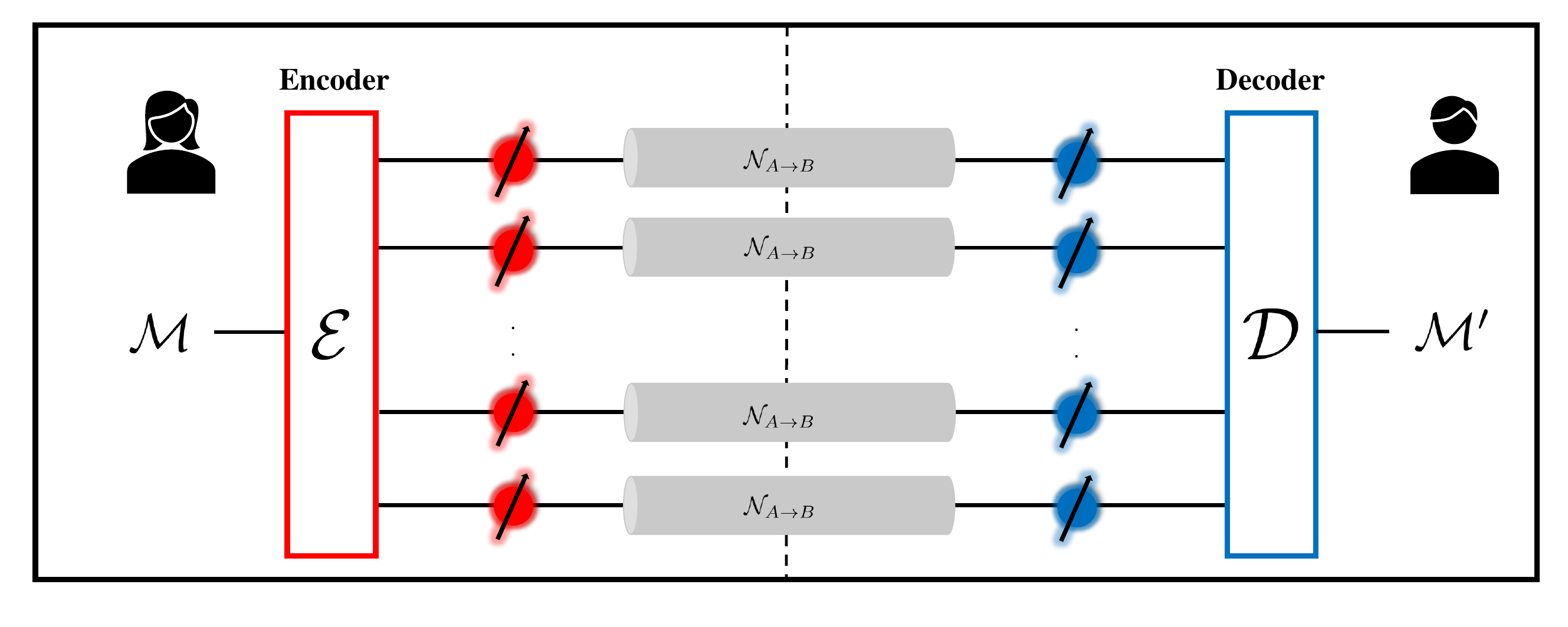}
      \caption{\textbf{Quantum communication:} Alice would like to send a quantum message to Bob over a noisy quantum channel $\mathcal{N}_{A\rightarrow B}$. Alice and Bob first agree on an encoder $\cE$ and decoder $\mathcal{D}$. The message $\mathcal{M}$ is then encoded into quantum states which are subsequently transmitted over several identical copies of the channel $\mathcal{N}_{A\rightarrow B}$. Bob receives the outputs of the channel and gets an estimate of the message $\mathcal{M}'$ employing the decoder. Using the channel multiple times along with a suitable encoding and decoding scheme, allows for reliable communication despite the noisy transmission. A natural question to ask in this setting is: what is the maximum rate at which Alice can reliably communicate with Bob? 
}
    \label{fig:overview}
\end{figure}

\section{Introduction}

A fundamental problem in information theory is to compute the maximum rate at which information can be reliably transmitted over a noisy channel. For channels described by classical physics, Shannon's seminal work~\cite{shannon1948mathematical} shows that every channel can be characterized by a single number\,---\,the capacity of the channel. More precisely, the channel capacity corresponds to the highest rate at which messages can be transmitted while ensuring that the error in the communication vanishes in the limit of long messages. In general, there is no closed analytical expression for the capacity of a channel but instead it is given in terms of an optimization problem. As such, efficient numerical methods to solve this optimization problem are sought after.

For channels governed by quantum physics, there exist multiple distinct settings where information must be transmitted over a noisy channel. Analogously to the classical regime, a quantum channel is given by any mapping that sends physical states to physical states. A graphical depiction is shown in Figure~\ref{fig:overview} and one may, for example, choose to transmit classical or quantum information, use pre-shared entanglement between the sender and the recipient to enhance the transmission rate, or look at minimizing the entropy gain between the output and input states of the channel. Hence, there exist many different relevant quantum channel capacities and\,---\,as with classical channels\,---\,these are typically given in terms of optimization problems which have to be solved numerically. 

\begin{figure}[t]
    \centering
    \includegraphics[width=.75\textwidth]{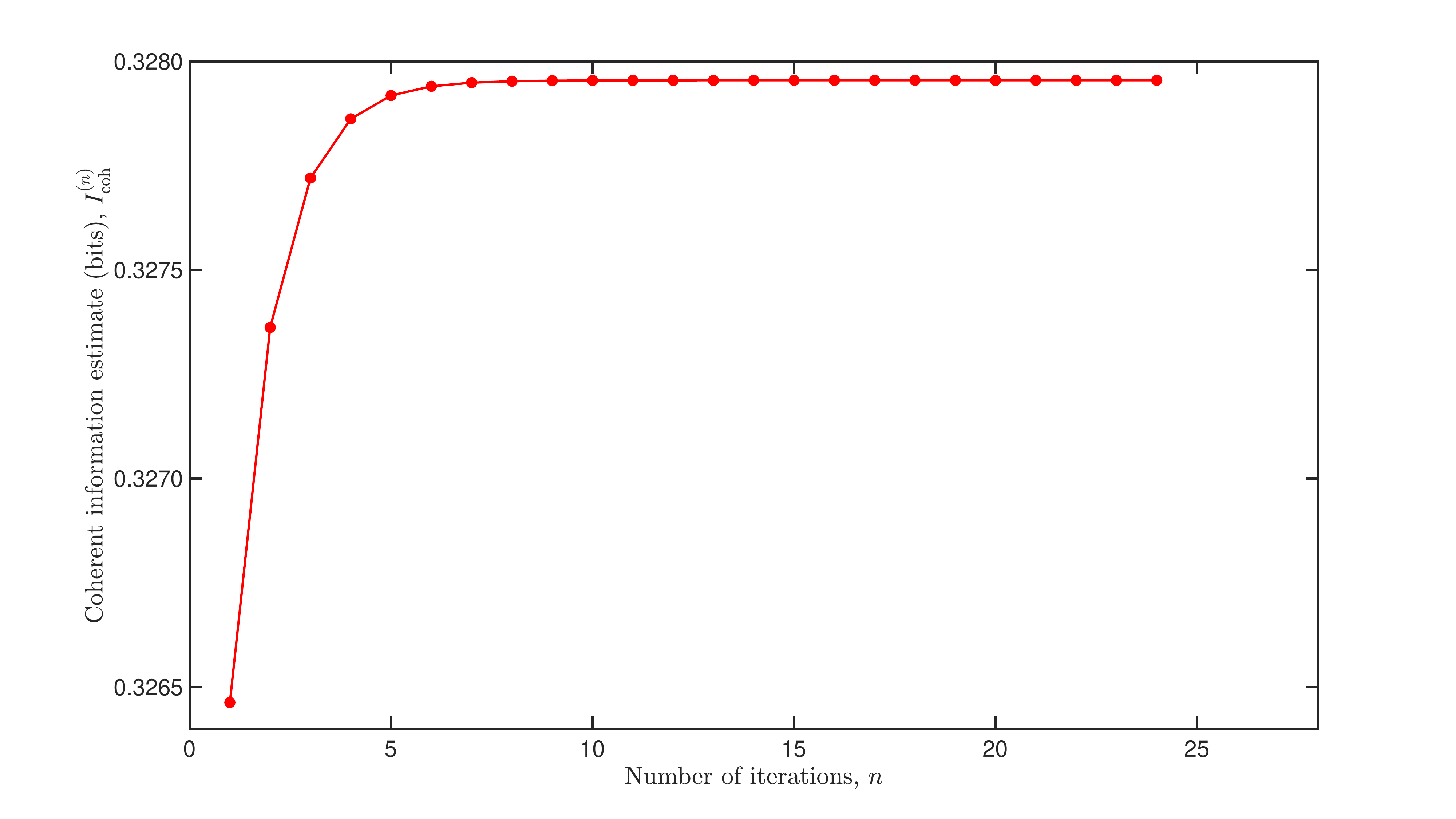}
    \caption{Convergence of the Blahut-Arimoto algorithm to the coherent information of the amplitude damping channel $\mathcal{E}^{AD}_{0.3}$ given in~\eqref{eq:amplitude_damping_channel_intro} with additive error in the estimate bounded by $\epsilon = 10^{-6}$ bits. The figure shows the lower bound on the coherent information in each iteration step $n$. The \textit{a priori} error criterion predicts $\mathcal{O}(10^6)$ iterations for the worst case but in practice, the algorithm is much faster requiring only $24$ iterations while ensuring that the additive error is bounded by $\epsilon$.}
    \label{fig:CI_amplitude_damping_intro}
\end{figure}

For concreteness, we focus in this introduction on the \textit{quantum capacity} of a channel which is related to the so-called \textit{coherent information} \cite{lloyd1997capacity,shor2002quantum,devetak2005private}. For a quantum channel $\cE$ the coherent information is defined as
\begin{align} \label{eq:defintion_ci_intro}
I_{\rm{coh}}(\cE) = \max_{\rho}\underbrace{S(\cE(\rho)) - S(\cE_c(\rho))}_{=I_{\rm{coh}}(\rho,\cE)}\,,
\end{align}
where the maximization is taken over input quantum states, the von Neumann entropy is denoted by $S(\rho)=- \Tr{\rho \log \rho}$, and $\cE_c$ is the complementary channel of $\cE$ (defined in Section~\ref{sec:notation}). As information is preserved in quantum physics, the complementary channel captures the notion that any information that is not transmitted to the recipient must have leaked to the environment. A channel is called \textit{less noisy} when the information that leaks to the environment is less than the information transmitted (for a precise definition, see \eqref{eq:less_noisy_condition}). In this case, the considered problem becomes convex \cite{watanabe2012private} and as in classical information theory, one could in principle use standard methods from convex optimization \cite{boyd2004convex} to solve the resulting problem. However, the gold standard to compute channel capacities in classical information theory are so-called \textit{Blahut-Arimoto} algorithms \cite{Blahut72,Arimoto72}. They have the following particularly strong properties:
 \begin{itemize}
 \item The algorithm is of an iterative form where each iteration step is of a very simple form given by a matrix exponential. Hence, no matrix derivatives have to be calculated in the update rules which is in contrast to, e.g., gradient based methods.
 \item Explicit \textit{a priori} and \textit{a posteriori} convergence guarantees are directly derived from entropic inequalities.
 \item The algorithm typically converges extremely fast in practice. Additionally, in certain cases, there is even exponentially fast convergence.
 \end{itemize}
In our work, we generalize Blahut-Arimoto algorithms to the quantum setting and show that aforementioned desirable properties are retained. For example, we show that the coherent information of less noisy channels can be approximated by iteratively improving on an estimate of the input state $\rho$ to the channel. The update rule for each iteration is given as
\begin{align}\label{eq:update_rule_intro}
\rho^{(n+1)} = \frac{1}{Z^{(n+1)}} \expB{\log\rho^{(n)} + \HE(\rho^{(n)})}, 
\end{align}
where $\HE(\sigma) = \cE_c^\dagger\log\cE_c(\sigma) - \cE^\dagger\log\cE(\sigma)$, $\cE^\dagger$ and $\cE_c^\dagger$ denote the adjoint channels (see e.g.~\cite{tomamichel_quantum_2015} for the definition) of $\cE$ and $\cE_c$, respectively, and $Z^{(n+1)}$ takes care of normalization. The estimate of the coherent information after the $n$-th iteration step is then given by $I_{\rm{coh}}^{(n)}=I_{\rm{coh}}(\rho^{(n+1)},\cE)$ and we prove that starting from an initial input state $\rho^{(1)}=\frac{1}{N}$ we have after
\begin{align}
\label{eq:iterations_bound}
\text{$n=\left\lceil\log\frac{N}{\varepsilon}\right\rceil$ iterations that $\big|I_{\rm{coh}}(\cE)-I_{\rm{coh}}^{(n)}\big|\leq\varepsilon$.}
\end{align}
Here, $N$ denotes the dimension of the input system of the channel. The number of iterations required in \eqref{eq:iterations_bound} corresponds to the worst-case and in practice, convergence is much faster for the channels we study. For example, we consider the qubit amplitude damping channel which is defined as follows
\begin{align}
\label{eq:amplitude_damping_channel_intro}
    \text{$\mathcal{E}^{AD}_{p}(\rho) =A_{0} \rho A_{0}^{\dagger}+A_{1} \rho A_{1}^{\dagger}$ with $A_{0} = | 0 \rangle\langle 0|+\sqrt{1-p}| 1\rangle\langle 1\vert, A_{1} = \sqrt{p}\vert 0\rangle\langle 1\vert$ for $p \in[0,1]$.}
\end{align}
Setting $p=0.3$, we compute $I_{\rm coh}(\cE^{AD}_{0.3})$ using the generalized Blahut-Arimoto algorithm up to an additive error of $\epsilon = 10^{-6}$ bits. We introduce an \textit{a posteriori} error criterion that allows us to find the error in our capacity estimate at each iteration and terminate the algorithm if this error is sufficiently small. When we use this criterion, the number of iterations is indeed much smaller than that required in \eqref{eq:iterations_bound} as shown in Figure~\ref{fig:CI_amplitude_damping_intro}.

In the following, we give an overview of our results showing that quantum Blahut-Arimoto algorithms can in fact be used to estimate several fundamental variants of channel capacities in quantum information theory.


\section{Overview of results}

\subsection{Quantum settings}

In addition to the coherent information of less noisy channels, other settings are of interest when one considers quantum channel capacities. For example, one may also ask the following questions:
\begin{enumerate}
    \item How much classical information can be sent per use of a noisy channel with classical input and quantum output? 
    \item How much information can be sent over a quantum channel in the presence of physical restrictions imposed by thermodynamics?
    \item How much quantum information can be sent over a quantum channel with entanglement assistance, i.e., with access to pre-shared entangled states between the two communicating parties?
\end{enumerate}
 These quantities are characterized by different capacities that are related to well-known quantities in quantum information theory, respectively the
 \begin{enumerate}
     \item Holevo quantity~\cite{holevo_capacity_1998,schumacher_sending_1997}
     \item Thermodynamic capacity~\cite{faist_thermodynamic_2018}
     \item Quantum mutual information~\cite{Bennett02}.
 \end{enumerate}  
 
We show that all the entropic optimization problems corresponding to these quantities as well as the coherent information discussed in the introduction can be cast into the same form and solved using an update rule similar to \eqref{eq:update_rule_intro} for suitable choice of $\mathcal{F}$. This generalization of the Blahut-Arimoto algorithm to the quantum setting provides the first efficient schemes for the calculation of the thermodynamic capacity of quantum channels, the coherent information of less noisy channels, and the quantum mutual information of quantum channels. 

\subsection{Algorithm speed and time complexity}
\label{sec:summary_res}
\begin{table}[ht]
\centering
\def\arraystretch{1.7}
\begin{tabular}{l|l|l}
  Channels & Quantity & Time Complexity \\
  \hline
  $X\rightarrow Y$ & Mutual information $I(\mathcal{E})$ & $\mathcal{O}\left(\frac{|X||Y|\log|X|}{\varepsilon}\right)$  \\
  $X\rightarrow B$ & Holevo quantity $\chi(\cE)$ & $\mathcal{O} \left(\frac{(|B|^3+|B|^2 |X|)\log |X|}{\varepsilon} + |X| |B|^3 \right)$ \\
  $A\rightarrow B$ & Thermodynamic capacity $T_\Gamma(\cE)$ & $\mathcal{O} \left(\frac{(|A|^3 + |A|^2 |B|^2+ |B|^3)\log |A|}{\varepsilon}\right)$ \\
   $A\rightarrow B$ & Coherent information $I_{\rm{coh}}(\cE)$ &$\mathcal{O} \left(\frac{\left(|A|^3 + |B|^3+ K^3+ |A|^2 (|B|^2+K^2)\right)\log |A|}{\varepsilon}\right)$ \\
   $A\rightarrow B$ & Quantum mutual information $I(\cE)$ & $\mathcal{O} \left(\frac{\left(|A|^3 + |B|^3+ K^3+ |A|^2 (|B|^2+K^2)\right)\log |A|}{\varepsilon}\right)$ 
\end{tabular}
\caption{Summary of the Blahut-Arimoto algorithms discussed in this work, with the overall asymptotic worst case complexity for an additive $\varepsilon$-approximation. $X$ and $Y$ refer to classical registers while $A$ and $B$ refer to quantum registers. The Kraus rank of $\cE$, i.e., the minimal number of Kraus operators required to represent $\cE$ (see e.g.~\cite{tomamichel_quantum_2015} for the definition of Kraus operators), is denoted by $K\leq |A| |B|$. The complexity for the classical case is the same as in the original works~\cite{Blahut72,Arimoto72}. For the coherent information of quantum channels, we require the channel $\cE$ to lie in the class of less noisy channels. }
\label{tbl:intro}
\end{table}

The number of iterations required by the generalized Blahut-Arimoto algorithm to obtain an \textit{a priori} error of $\varepsilon$ is $\mathcal{O}\left(\varepsilon^{-1}\log N\right)$, where $N$ denotes the input dimension of the channel. To find the worst-case time complexity of the full algorithm that estimates the capacity up to an additive error $\varepsilon$, one can simply multiply the required number of iterations with the complexity of the update rule given in \eqref{eq:update_rule_intro} for appropriate choice of $\mathcal{F}$. The full complexity of the algorithm for an $\varepsilon$-estimation of the capacity is given in Table~\ref{tbl:intro}.


\subsection{Additional speed-ups}

One can also construct an accelerated version of the generalized Blahut-Arimoto algorithm similar to the classical case~\cite{Matz04,Naja09}. The accelerated version lowers the number of iterations by a constant factor in the worst case compared to the standard version and hence, the asymptotic results given Table~\ref{tbl:intro} are not affected by this. It also gives rise to heuristics that provide a significant speed-up of the algorithm in practical instances as shown in Figures~\ref{fig:CQ_random_states}, \ref{fig:TC_amplitude_damping}, \ref{fig:CI_amplitude_damping} and \ref{fig:MI_amplitude_damping}. 

The \textit{a priori} bound on the error corresponds to the worst case and the number of iterations required is extremely conservative for most problems. Hence, we adapt classical techniques~\cite{Vontobel03} to find \textit{a posteriori} error bounds on the capacity estimate at each iteration of the algorithm (see also~\cite{li2019computing}). This error bound can then be used as a termination criterion for our numerics as described in Sections~\ref{sec:numerics-cq}, \ref{sec:numerics-thermo}, \ref{sec:numerics-ci} and \ref{sec:numerics-mi}. In fact, even if we use an alternative algorithm or heursitics, the \textit{a posteriori} criterion gives us guarantees on the error of our estimate compared to the actual capacity.

We notice that for the channels and capacities we consider, the \textit{a posteriori} error is achieved in significantly fewer iterations compared to the \textit{a priori} requirement. One of the explanations for this speedup is that the convergence of our algorithm can be exponentially fast in certain cases. Specifically, when certain entropy inequalities used by our convergence proof are strict inequalities, one can prove (see Proposition~\ref{prop:exp_convergence} for more details) that the error in the capacity estimate after the $t^{\rm th}$ iteration for a channel with input dimension $N$ is bounded by $(1-\delta)^t\log N$ for some $\delta\in (0, 1]$. 

\subsection{Comparison with previous work}

Previous works have considered other methods to approximate the solution of entropic optimization problems in quantum information theory~\cite{nagaoka1998algorithms,Nagaoka01,Shor03,Hayashi05, kato2009computational, tanuma2011revisiting, Sutter16,Fawzi18,Fawzi2019,Gour10,Girard14}. Here, we provide the first fully quantum version for Blahut-Arimoto type algorithms, give convergence guarantees and demonstrate numerically that these algorithms outperform all previously known algorithms for the specific optimization tasks. In particular, calculating the classical capacity of classical-quantum channels has been the subject of several previous works. Nagaoka~\cite{nagaoka1998algorithms} proposed a Blahut-Arimoto type algorithm to compute the Holevo quantity but did not provide an analysis of the convergence and the time complexity. 
A provable convergence bound for the estimation of the Holevo quantity was previously given in~\cite{Sutter16} by using an algorithm based on convex programming duality and smoothing techniques and has a time complexity of
\begin{align}
\text{$\mathcal{O}\left(\varepsilon^{-1}\max\{|X|,|B|\}|B|^3\sqrt{\log|X|}\right)$ versus our $\mathcal{O} \left(\varepsilon^{-1}(|B|^3+|B|^2 |X|)\log |X| + |X| |B|^3 \right)$.}
\end{align}
This suggests that our generalized Blahut-Arimoto algorithm is faster for the regime $|B|\gg|X|$ but slightly slower for $|X|\gg|B|$ in the worst case. However, our numerical results suggest that Blahut-Arimoto type algorithms converge much faster in practice than the algorithm given in~\cite{Sutter16}. Further, the techniques used in~\cite{Sutter16} could not be generalized to fully quantum settings. For other algorithms that compute the Holevo quantity, there is often no explicit complexity analysis given. In particular, in~\cite{kato2009computational,tanuma2011revisiting}, an algorithm to approximate the Holevo quantity of arbitrary quantum channels is provided. However, no convergence guarantees are given and numerics suggest that the runtime of the algorithm is quite high in practice~\cite{kato_voronoi_2008}. 

We compare our numerics with the best known runtimes given in~\cite{Fawzi18}, whose algorithm is based on semi-definite approximations of the matrix logarithm~\cite{Fawzi2019}. Our numerics suggest that for the specific optimization problems considered here, Blahut-Arimoto type algorithms converge faster (Section~\ref{sec:numerics-cq}). We emphasize that while the advantage of Blahut-Arimoto type algorithms might be explained by their conformity to the specific structure of entropy optimization problems, the method presented in~\cite{Fawzi18} is applicable for a wider range of problems.

The rest of the paper is structured as follows. We first give the general structure and convergence proof of Blahut-Arimoto alternating optimization algorithms in Section~\ref{sec:general} and then discuss in more detail the classical-quantum case (incorporating the classical case) in Section~\ref{sec:cq}, the thermodynamic capacity of quantum channels in Section~\ref{sec:thermo}, the coherent information of less noisy channels in Section~\ref{sec:ci} and the mutual information of quantum channels in Section~\ref{sec:mi}. Numerical results are provided for the different cases in Sections~\ref{sec:numerics-cq}, \ref{sec:numerics-thermo}, \ref{sec:numerics-ci} and \ref{sec:numerics-mi}, respectively. Finally, we discuss some conclusions and give an outlook (Section~\ref{sec:conclusion}).


\section{Structure of quantum Blahut-Arimoto algorithms}\label{sec:general}

\subsection{Notation} \label{sec:notation}

We label Hilbert spaces (which are assumed to be finite-dimensional) with capital letters e.g. $A$, $B$, and denote their dimension by $|A|$, $|B|$, and so on. The set of density operators on a system $A$, i.e., positive semi-definite matrices $\rho_A$ with $\Tr{\rho_A}=1$, is denoted $\cD(A)$. Whenever we work with a single system, we may omit the sub-index referring to the system. A quantum channel $\cE_{A \rightarrow B}$ is a linear completely positive trace-preserving map that maps density operators on a system $A$ to density operators on a system $B$. For simplicity, we may write $\cE$ instead of $\cE_{A \rightarrow B}$ if the input and output systems of the channel are clear from the context. The complementary channel $\cE^c_{A\to E}$ of a channel $\cE_{A \rightarrow B}$ is defined via taking the partial trace over $B$ of the Stinespring dilation $U_{A\to BE}$ of $\cE_{A\to B}$. 

For a density operator $\rho$, the von Neumann entropy is defined as $S(\rho)=- \Tr{\rho \log \rho}$ and for density operators $\rho,\sigma$ the quantum relative entropy is defined as
\begin{align}\label{eq:rel_entropy}
D(\rho \| \sigma)= \begin{cases}
\Tr{\rho \left(\log \rho -\log \sigma \right)} &\text{if $ \sigma \gg \rho  $}\\
\infty &\text{otherwise}\,,
\end{cases}
\end{align}
where the notation $\sigma \gg \rho$ denotes that the kernel of $\sigma$ is a subset of the kernel of $\rho$ (i.e., $\ker(\sigma) \subseteq \ker(\rho)$), and where we take the logarithm only on the support of the argument. We also work with discrete probability distributions, which we represent as vectors $\lambda=[\lambda_1,\dots,\lambda_m]$ with $\sum_i \lambda_i =1$ or alternatively as diagonal matrices with entries $\lambda_1,\dots,\lambda_m$ i.e., $\rho_\lambda=\sum_{i=1}^m \lambda_i \ket{i}\bra{i}$ in some pre-defined basis. The definition of the von Neumann entropy and the quantum relative entropy then simplifies for probability distributions to the Shannon entropy and the Kullback-Leibler divergence, respectively.


\subsection{Blahut-Arimoto type}

In the following we study a special entropic type of alternating optimization algorithms, called Blahut-Arimoto algorithms~\cite{Blahut72,Arimoto72}\footnote{In the broader context of numerical optimization, Blahut-Arimoto algorithms can be viewed as a proximal point method that maximizes a function iteratively with a penalty term if the current guess moves away from the previous guess~\cite{Matz04,Naja09}.}. The channel capacities we are interested in are typically given as a convex optimization problem over input states. One may write an extension function $J$ in two variables such that the maximization over both variables gives back the capacity of the channel. Performing the maximizations iteratively leads to an algorithm of the following form.

\begin{algorithm}[H] 
  \caption{Blahut-Arimoto algorithm: Iterative double optimization over density operators}
  \label{algo:general_AB}
   \begin{algorithmic}[1]
   \State Inputs: \begin{itemize}
   \item Initial guess $\rho^{(1)}_A \in \cD(A)$ with full support, i.e.,   $\rho^{(1)}_A>0$
   \item Function $J_{\gamma}: \cD(A) \times \cD(B) \mapsto \mathbb{R}$ with a coefficient $\gamma >0$
   \item \text{Update relations} $\mathcal{F}_1: \cD(A) \mapsto \cD(B)$ and $\mathcal{F}_2: \cD(B) \mapsto \cD(A)$
   \item Number of iteration steps $n$
   \end{itemize}
   \For{$t \in \{1,2,\dots, n\}$}
     \State $\sigma_B^{(t)} = \mathcal{F}_1\left(\rho_A^{(t)}\right)$
      \State $\rho_A^{(t+1)} = \mathcal{F}_2\left(\sigma_B^{(t)}\right)$
      \EndFor
   \State Outputs: $\rho_A^{(n+1)}$, $C(n)= J_{\gamma}\left(\rho_A^{(n+1)}, \sigma_B^{(n)}\right)$, 
   
   where $C(n)$ should approximate $C^\star =\max_{\rho_A,\sigma_B}J_{\gamma}(\rho_A,\sigma_B)$ for $n \rightarrow \infty \, $.
   \end{algorithmic}
\end{algorithm}

\begin{definition}[Blahut-Arimoto]\label{defi:AB_form}
Algorithm~\ref{algo:general_AB} together with the following conditions on $J$, $\mathcal{F}_1$ and $\mathcal{F}_2$ defines a quantum Blahut-Arimoto algorithm: For $\gamma>0$ and density operators $\sigma \gg \rho$,
\begin{align} \label{eq:form_J_gamma_general}
J_\gamma(\rho,\sigma)=-\gamma \, D( \rho||\sigma) + \Tr{\rho\HE(\sigma)} \in \mathbb{R}\,,
\end{align}
where $\HE$ is a Hermitian matrix valued super-operator on density operators such that $\Tr{\rho \HE(\sigma)}$ is continuous in $\sigma$ for $\sigma \gg \rho$. The update rules are then given by
\begin{align} \label{eq:update_rules}
\mathcal{F}_1(\rho) = \argmax_{\sigma \textnormal{ with } \sigma \gg \rho} J_\gamma(\rho,\sigma) \, , \quad  \mathcal{F}_2(\sigma) = \argmax_{\rho \textnormal{ with } \sigma \gg \rho} J_\gamma(\rho,\sigma)\, .
\end{align}

\end{definition}

Now, under certain conditions we can find analytic expressions for the optimizers in~\eqref{eq:update_rules}.\footnote{We can choose any optimizer in the above optimizations if there is not a unique one.}

\begin{lemma}[Update rules]\label{lem:update_rules}
For quantum Blahut-Arimoto algorithms with 
\begin{align} \label{eq:cond_on_F}
\Tr{\rho\left\{\HE\left(\sigma\right) - \HE\left(\rho\right)\right\}}\leq \gamma D(\rho||\sigma)\quad\text{for density operators $\sigma \gg \rho$,}
\end{align}
the update rules appearing in Algorithm~\ref{algo:general_AB} are of the form
\begin{align} 
&\mathcal{F}_1(\rho) =\rho \label{eq:update_rule_convergence_lemma_1} \, ,\\
&\text{$\mathcal{F}_2(\sigma)=\frac{1}{Z(\sigma)}\expB{\log\, \sigma+\frac{1}{\gamma} \HE\left(\sigma\right)}$ with $Z(\sigma)=\Tr{\expB{\log\, \sigma+\frac{1}{\gamma}\HE\left(\sigma\right)}}$}
\label{eq:update_rule_convergence_lemma_2}
\end{align}
for density operators $\rho$ and $\sigma >0$. 
\end{lemma}

\begin{proof} 
By~\eqref{eq:cond_on_F} together with the definition in~\eqref{eq:form_J_gamma_general}, we find
\begin{align}
J_\gamma(\rho,\rho) = \Tr{\rho \HE\left(\rho\right)}\geq \Tr{\rho \HE\left(\sigma\right)} - \gamma D(\rho||\sigma) = J_\gamma(\rho,\sigma)\quad\text{for density operators $\sigma \gg \rho$.}
\end{align}
Hence, the fist update rule is given by 
\begin{align}
\mathcal{F}_1(\rho) =  \argmax_{\sigma \gg \rho}J_\gamma(\rho,\sigma) = \rho \, .
\end{align}
For the second update rule we assume that $\sigma >0$ and hence the optimizer $\rho$ can be chosen from the set of density operators. The form of the optimizer follows directly from Gibbs' variational principle (Lemma~\ref{lem:Gibbs}). To see this, we write for $\sigma >0$,
\begin{align}
    \argmax_\rho J_\gamma(\rho, \sigma)  =  \argmax_\rho - \Tr{\rho \log \rho} +  \Tr{\rho \left\{ \log \sigma + \frac1\gamma \HE(\sigma) \right\}}\, .
\end{align}
Setting $\omega = \rho$ and $H = \log \sigma + \frac1\gamma \HE(\sigma)$ in Lemma~\ref{lem:Gibbs} leads to the claimed form of the second update rule.
\end{proof}\newline

The next theorem shows convergence of quantum Blahut-Arimoto algorithms under some technical assumptions (which will be satisfied for the applications in Sections~\ref{sec:cq}-\ref{sec:ci}).

\begin{theorem}[Convergence]\label{thm:convergence}
For quantum Blahut-Arimoto algorithms together with a strictly positive definite initial state $\rho^{(1)}>0$ on a system $A$ and 
\begin{align} \label{eq:bounding_condition}
0 \leq \Tr{\rho\left\{\HE\left(\sigma\right) - \HE\left(\rho \right)\right\} }  \leq \gamma D(\rho||\sigma)\quad\text{for density operators $\sigma \gg \rho$,}
\end{align}
we have that $C(n)$ of Algorithm~\ref{algo:general_AB} is monotonically increasing and converges for $n\rightarrow \infty$ to
\begin{align}
C^{\star}= \max_{\rho, \sigma \textnormal{ with } \sigma \gg \rho} J_\gamma(\rho, \sigma)
\end{align}
with the approximation error bounded as 
\begin{align}
|C^\star - C(n)| \leq \frac{ \gamma D\left(\rho^\star\middle\|\rho^{(1)}\right)}{n} \, ,
\end{align} 
where $\rho^\star$ is an optimizer (possibly not unique) that achieves the capacity $C^\star$. If $\rho^{(1)}$ is chosen to be the maximally mixed state, the approximation error is bounded as $ |C^\star - C(n)| \leq \frac{\gamma\log|A|}{n}$.
\end{theorem}

\begin{proof}
Let $t\in \mathbb{N}$ and let $\rho^{(t)} \in \cD(A)$ be a density operator in the $t$-th iteration step of Algorithm~\ref{algo:general_AB}. Note that $\rho^{(t)} >0$ for all $t$, since the exponentiation of a matrix in the update rule~\eqref{eq:update_rule_convergence_lemma_2} ensures full support.  First, using the update rules given in~\eqref{eq:update_rule_convergence_lemma_1} and in~\eqref{eq:update_rule_convergence_lemma_2}, and setting $Z^{(t+1)}=Z(\rho ^{(t)})$, we find 
 \begin{align}
   C(t)&= J_\gamma(\rho^{(t+1)},\sigma^{(t)}=\rho^{(t)})  \\
    &=   \, \Tr{  \rho^{(t+1)} \left[ - \gamma \log \rho^{(t+1)} +\gamma \log \rho^{(t)} + \, \HE\left(\rho^{(t)}\right)\right]}\\
    &=  \, \Tr{  \rho^{(t+1)} \left[ - \gamma \left\{\log \rho^{(t)}  + \frac{1}{\gamma} \HE\left(\rho^{(t)}\right) - \log Z^{(t+1)}\right\} +\gamma \log \rho^{(t)} + \, \HE\left(\rho^{(t)}\right)\right]}\\
     &=   \,  \gamma \, \log Z^{(t+1)}\,.
    \end{align}
Let $\rho^{\star}= \argmax_{\rho} J_\gamma(\rho,\rho)$, and note that such a $\rho^{\star}$ always exists because $J_\gamma(\rho,\rho)$ is continuous by Definition~\ref{defi:AB_form} and we optimize over the compact set of all density operators. Further, we have $C^{\star}=J_\gamma(\rho^\star,\rho^\star)=\Tr{\rho^\star \, \HE(\rho^\star)}$.
Using $C(t)= \gamma \, \log Z^{(t+1)}$, we derive an upper bound on the additive error at the iteration step $t$
\begin{align}
    \Tr{\rho^\star(\log\rho^{(t+1)} - \log\rho^{(t)})} 
   & = -\frac{1}{\gamma}C(t) + \frac1\gamma \Tr{\rho^\star \HE(\rho^{(t)}) } \label{eq:conv_proof_eq1}  \\
    &= \frac{1}{\gamma}\left (C^\star - C(t) + \Tr{\rho^\star \left\{\HE\left(\rho^{(t)}\right) - \HE\left(\rho^\star \right)\right\}}\right)  \label{eq:conv_proof_eq2}\\
    &\geq \frac{1}{\gamma}|C^\star - C(t)|  \, ,\label{eq:conv_proof_eq3}
\end{align}
where we used the assumption~\eqref{eq:bounding_condition} of the lemma in the last step and that $C^{\star}\geq C(t)$, since $C^{\star}=J_\gamma(\rho^\star,\rho^\star)=\max_{\rho, \sigma \textnormal{ with } \sigma \gg \rho} J_\gamma(\rho, \sigma)$ is the maximum value that $J_\gamma$ can achieve.
The sum over the additive error terms is upper bounded by a telescopic sum, which can itself be upper bounded as follows
\begin{align} 
\sum_{t = 1}^k \Tr{\rho^\star\left(\log\rho^{(t+1)} - \log\rho^{(t)}\right)}&=  \Tr{\rho^\star\left(\log\rho^{(k+1)} - \log\rho^{(1)}\right)}\\
&=  D\left(\rho^\star\middle\|\rho^{(1)}\right) - D\left(\rho^\star\middle\|\rho^{(k+1)}\right) \\
& \leq D\left(\rho^\star\middle\|\rho^{(1)}\right) \, ,
\end{align}
where we used the positivity of the quantum relative entropy in the last inequality. We conclude
\begin{align}\label{eq:error_bound_in_proof}
\sum_{t=1}^\infty |C^\star - C(t)| \leq \gamma D\left(\rho^\star\middle\|\rho^{(1)}\right)\,,
\end{align}
and since $D\left(\rho^\star\middle\|\rho^{(1)}\right)$ is finite for $\rho^{(1)} \in \cD(A) $ with full support, $C(n)$ converges to $C^{\star}$ for $n\rightarrow \infty$. Moreover, $C(n)$ increases monotonically in $n$ by construction of the updated states as optimizers in Lemma~\ref{lem:update_rules}. Together with~\eqref{eq:error_bound_in_proof}, this implies the error bound stated in the theorem.
\end{proof}\newline

The value of the parameter $\gamma$ appearing in Definition~\ref{defi:AB_form} corresponds to accelerated versions of the Blahut-Arimoto algorithm and hence, $\gamma$ should be chosen as small as possible. However, to have provable convergence, $\gamma$ is lower bounded by the requirement in Theorem~\ref{thm:convergence} that~\eqref{eq:bounding_condition} is satisfied. The standard algorithm uses the data processing inequality for the quantum relative entropy to find a $\gamma$ satisfying~\eqref{eq:bounding_condition}, whereas the accelerated algorithm uses relative entropy coefficients similar to~\eqref{eq:contraction-coefficient} to establish tight lower bounds on $\gamma$. Note that calculating these coefficients is itself a difficult optimization problem and may be intractable for many channels. However, as a heursitic, we may still introduce an adaptive acceleration parameter $\gamma^{(t)}$ at each iteration, similar to the classical case~\cite{Matz04}.

\begin{rmk}[Acceleration heuristic]
For $D(\rho^{(t)}||\rho^{(t-1)})\neq 0$, define an adaptive acceleration parameter
\begin{align}\label{eq:adaptivegamma}
    \gamma^{(t+1)}=\frac{\Tr{\rho^{(t)}\left\{\HE\left(\rho^{(t-1)}\right) - \HE\left(\rho^{(t)} \right)\right\}}}{D\left(\rho^{(t)}\middle\|\rho^{(t-1)}\right)} \, .
\end{align}
This adaptive acceleration parameter may be used when strict lower bounds on $\gamma$ cannot be calculated. The proof of convergence from Theorem~\ref{thm:convergence} no longer holds since the capacity estimate does not necessarily increase monotonically when the adaptive acceleration method is used. Nevertheless, we find in practice that the adaptive acceleration heuristic provides a significant speed-up (see Figures~\ref{fig:CQ_random_states}, \ref{fig:TC_amplitude_damping}, \ref{fig:CI_amplitude_damping} and \ref{fig:MI_amplitude_damping}).
\end{rmk} 

Finally, regardless of whether we use the adaptive acceleration parameter or otherwise, the following proposition allows us to terminate the algorithm when the \textit{a posteriori} error is sufficiently small (see also~\cite{li2019computing}).

\begin{prop}[Termination criteria]\label{prop:dynamic_termination}
For quantum Blahut-Arimoto algorithms that satisfy~\eqref{eq:bounding_condition} let us denote the largest eigenvalue of $\HE(\rho^{(t)})$ by $\alpha^{(t)}_{\max}$ for $t \in \mathbb{N}$. Then, we have  
\begin{align} \label{eq:prop_posteriori_criteria}
C(t) \leq C^\star \leq \alpha^{(t)}_{\max}
\end{align}
with $C^{\star}= \max_{\rho, \sigma \textnormal{ with } \sigma \gg \rho} J_\gamma(\rho, \sigma)$. Hence, in particular, $|C^\star - C(t)| \leq \alpha^{(t)}_{\max} - C(t) $  provides a bound on the \textit{a posteriori} error.
\end{prop}

\begin{proof}
For any $t \in \mathbb{N}$, we have $\Tr{\rho^\star(\HE(\rho^{(t)}) - \HE(\rho^\star))}  \geq 0$ from~\eqref{eq:bounding_condition}, where $\rho^\star$ is the optimizer (possibly not unique) that achieves the capacity $C^\star$.  Therefore, there must exists at least one eigenvalue $\alpha_i^{(t)}$ of $\HE(\rho^{(t)})$ that satisfies $\alpha_i^{(t)} \geq \Tr{\rho^\star \, \HE(\rho^\star)} = C^\star$. In particular, $\alpha_{\max}^{(t)}\geq C^\star$. Noting that $C(t) \leq C^\star$ finishes the proof.
\end{proof}\\
\newline
In particular, if $\rho^{(t)}$ converges to a full rank optimizer $\rho^{\star}$, then we may take the logarithm on both sides of \eqref{eq:update_rule_convergence_lemma_2} to show that $\lim_{t\rightarrow\infty}\alpha^{(t)}_{\max} = C^\star$. This shows that the \textit{a posteriori} criterion is indeed achieved for some sufficiently large $t$ for full rank optimizers. Next, we show in the following proposition that under certain technical conditions, the speed of convergence is exponential.

\begin{prop}[Exponential convergence]\label{prop:exp_convergence}
For quantum Blahut-Arimoto algorithms satisfying~\eqref{eq:bounding_condition} with $a\in (0, \gamma]$ such that
\begin{align} \label{eq:assumption_exp_convergence_prop}
\inf_{\rho\neq\sigma}\frac{\Tr{\rho\left\{\HE\left(\sigma\right) - \HE\left(\rho\right)\right\}}}{D(\rho||\sigma)} = a \quad\text{for density operators $\sigma \gg \rho$}\,,
\end{align}
there exists a unique optimizer $\rho^\star$ with $C^{\star}= J_\gamma(\rho^\star, \rho^\star)$, and the algorithm converges exponentially fast. That is, for all $t\in \mathbb{N}$ we have
\begin{align}\label{eq:exp_conv_2}
C^\star - C(t+1) &\leq \gamma\left(1 - \delta\right)^{t}D(\rho^\star||\rho^{(1)})\,,
\end{align}
where $\delta=\frac{a}{\gamma}\in(0 ,1]$. Moreover, when $\rho^{(1)}$ is chosen to be the maximally mixed state, we obtain $C^\star - C(t+1) \leq \gamma(1-\delta)^t\log|A|$. 
\end{prop}
\begin{proof}
We first establish the uniqueness of the optimizer. From~\eqref{eq:conv_proof_eq1} and~\eqref{eq:conv_proof_eq2} in the proof of Theorem~\ref{thm:convergence}, we have
\begin{align}
    D(\rho^\star||\rho^{(t)}) - D(\rho^\star||\rho^{(t+1)}) &=  \frac{1}{\gamma}\left (C^\star - C(t) + \Tr{\rho^\star \left\{\HE\left(\rho^{(t)}\right) - \HE\left(\rho^\star \right)\right\}}\right) \label{eq:conv_proof_eq2_local}\\
    &\geq \frac{a}{\gamma}D(\rho^\star||\rho^{(t)}) \\
    \implies D(\rho^\star||\rho^{(t+1)}) &\leq \left(1 - \frac{a}{\gamma}\right)D(\rho^\star||\rho^{(t)}) \, . \label{eq:exp_conv_implication}
\end{align}
Note that~\eqref{eq:exp_conv_implication} implies that there is a unique optimizer $\rho^\star$ since this inequality holds for each $t$. Indeed, to see this, we start by noting that~\eqref{eq:exp_conv_implication} implies $\lim_{t\rightarrow\infty} D(\rho^\star||\rho^{(t)}) =0$. Then, Pinsker's inequality~\cite[Theorem 5.38]{watroustheory} implies that also $\lim_{t\rightarrow\infty}  \norm{\rho^{(t)}-\rho^{\star}}_1 =0$. Since the trace distance is a norm, we find $\lim_{t\rightarrow\infty} \rho^{(t)} = \rho^{\star}$. Since this must be true for any optimizer $\rho^{\star}$, the optimizer must be unique.
From~\eqref{eq:bounding_condition} and~\eqref{eq:conv_proof_eq2_local}, we find
\begin{align}
    C^\star - C(t) &\leq \gamma(D(\rho^\star||\rho^{(t)}) - D(\rho^\star||\rho^{(t+1)})) \\
    &\leq \gamma D(\rho^\star||\rho^{(t)}) \, .
\end{align}
Noting that this argument holds for all $t$ finishes the proof.
\end{proof}

\section{Holevo quantity of classical quantum channels}\label{sec:cq}

\subsection{Definitions}

Here we consider how to compute the capacity of cq channels, as given by the Holevo quantity. A cq channel can be described by a set of input-output pairs $\{(x,\tau_x)\}_{x \in \{1,2,\dots,N\}}$, i.e., the channel takes $x \in \{1,2,\dots,N\}$ as an input and provides the quantum states $\tau_x \in \cD(B)$ as outputs. Alternatively, we may consider cq channels as quantum channels $\cE_{\X \rightarrow \Y}$ with $|X|=N$, defined by the completely positive trace-preserving mapping $\cE_{\X \rightarrow \Y}: \rho_X \mapsto \sum_x  \bra{x} \rho_X \ket{x} \, (\tau_x)_B$. Restricting the quantum channel to the classical input states $\ket{x}\bra{x}_X$ then gives us back the original cq channel. For an input distribution vector $\lambda$, the output of the cq channel corresponds to
\begin{align}
\cE(\rho_\lambda) = \sum_x \lambda_x \, \cE(\ket{x}\bra{x})=  \sum_x \lambda_x \, \tau_x\,,
\end{align}
where $\lambda_i$ denotes the $i$-th component of the probability vector $\lambda$ and $\rho_\lambda = \sum_k\lambda_k\ket{k}\bra{k}$. As shown by Holevo, Schumacher and Westmoreland~\cite{holevo_capacity_1998,schumacher_sending_1997}, the classical capacity of a cq channel is given by the Holevo quantity
\begin{align}\label{cq_capacity}
\text{$\chi(\cE) = \max_{\lambda}\underbrace{\sum_i\lambda_i \, \Tr{\tau_{\cE,i}\left\{\log \tau_{\cE,i} - \log \cE(\rho_\lambda)\right\}}}_{=I(\lambda, \cE)}$ \, \, \, \, with $\tau_{\cE,i}=\cE(\ket{i}\bra{i})$,}
\end{align}
where $I(\lambda, \cE)$ is the mutual information between the input and output registers.

\subsection{Blahut-Arimoto algorithm}

The following algorithm to compute the Holevo quantity is a quantum Blahut-Arimoto algorithm and was first proposed by Nagaoka~\cite{nagaoka1998algorithms}. Similar to the classical case~\cite{Matz04}, we will show that one can accelerate the convergence of Blahut-Arimoto type algorithms. To do so, we will use the contraction coefficient of relative entropy~\cite{Ruskai99}, given for a channel $\cE$ as
\begin{align}\label{eq:contraction-coefficient}
\eta_{\rm con}(\cE)=\sup_{\rho\neq\sigma} \frac{D\left(\cE(\rho)\middle\|\cE(\sigma)\right)}{D\left(\rho\middle\|\sigma\right)} \in [0, 1]\ .
\end{align}  
Note that $0 \leq \eta_{\rm con}(\cE) \leq 1$ by the nonnegativity of relative entropy and by the data processing inequality. We now define the following two variable extension of the mutual information for $\gamma \in [\eta_{\rm con}(\cE),1]$.
\begin{align}
   J_\gamma(\lambda, \mu, \cE)  &= - \gamma D(\lambda ||\mu) + \sum_i\lambda_i \Tr{\tau_{\cE,i}\left\{\log \tau_{\cE,i} - \log \cE(\rho_\mu)\right\}} \\
   &= I(\lambda,\cE) - \gamma D(\lambda ||\mu) + D(\cE(\rho_\lambda) || \cE(\rho_\mu)) \, ,\label{eq:J_cq_expression}
\end{align}
where $\rho_\lambda = \sum_k\lambda_k\ket{k}\bra{k}$, $\rho_\mu = \sum_k\mu_k\ket{k}\bra{k}$ (with $\rho_\mu \gg \rho_\lambda$) and $\tau_{\cE,i}=\cE(\ket{i}\bra{i})$. Note that $\gamma = 1$ corresponds to the standard algorithm and if it is known that $\eta_{\rm con}(\cE) < 1$, then one may accelerate the algorithm. The condition $\eta_{\rm con}(\cE) < 1$ implies that $\forall \rho\neq\sigma,\ D(\cE(\rho)||\cE(\sigma)) < D(\rho||\sigma)$ i.e. the map strictly decreases the relative entropy. We then bring $J_\gamma$ into the form
\begin{align}
 J_\gamma &= \ - \gamma  \,\Tr{\rho_\lambda \log \rho_\lambda} +\Tr{ \rho_\lambda (\gamma \log \rho_\mu +\HE(\rho_\mu))} \, \textnormal{with} \label{J_gamma_cq_standard_form}\\
 \HE(\rho_\mu) &= \sum_i \ket{i}\bra{i} \,  \Tr{\tau_{\cE,i}\left(\log \tau_{\cE,i} - \logs \cE(\rho_\mu)\right)} \, . \label{eq:defi_F_cq}
\end{align}
Note that $\Tr{\rho_\lambda \HE(\rho_\mu) }$ is continuous for all $\rho_\mu \gg \rho_\lambda$, since for $\lambda_i \neq 0$, we have $\mu_i \neq 0$ and hence $\ker {\left[ \cE(\rho_\mu)\right]}=\ker{\left[ \sum \mu_k  \tau_{\cE,k} \right]}\subset \ker{\left[ \tau_{\cE,i}\right]}$. Further, a calculation shows 
\begin{align}
&\Tr{\rho_{\lambda}\left\{\HE(\rho_{\mu}) - \HE(\rho_{\lambda})\right\}}=D(\cE(\rho_{\lambda})||\cE(\rho_{\mu}))  \textnormal{ and hence } \label{eq:useful_form_bounding_cond_cq}\\
&0 \leq \Tr{\rho_{\lambda}\left\{\HE(\rho_{\mu}) - \HE(\rho_{\lambda})\right\}} \leq \gamma D(\rho_{\lambda}||\rho_{\mu})\,, \label{eq:requirement_cq}
\end{align} 
where we used the nonnegativity of the quantum relative entropy in the last inequality (together with $\gamma \in [\eta_{\rm con}(\cE),1]$). 
From Lemma~\ref{lem:update_rules} whose requirement is satisfied by~\eqref{eq:requirement_cq}, we find a double optimization form of the capacity $\chi(\cE)$
\begin{align}
\max_{\lambda, \mu \textnormal{ with } \rho_\mu \gg \rho_\lambda} J_\gamma(\lambda, \mu, \cE)&=\max_{\lambda} J_\gamma(\lambda, \lambda, \cE)=\max_{\lambda} I(\lambda, \cE) =\chi(\cE)\,.
\end{align}

Performing the two maximizations in $\max_{\lambda, \mu} J_\gamma(\lambda, \mu, \cE)$ iteratively, leads to the following algorithm (see Lemma~\ref{lem:update_rules} for the form of the update rules), which provably converges  to the Holevo quantity by Theorem~\ref{thm:convergence}. A detailed analysis of the time complexity can be found in Appendix~\ref{app:time_compl_cq}. 

\begin{algorithm}[H]\
  \caption{Accelerated Blahut-Arimoto algorithm for the Holevo quantity}
  \label{algo:AB_cq_acc}
   \begin{algorithmic}[1]
    \State Inputs: cq channel $\cE_{\X\rightarrow \Y}$ (given as a lookup table whose $i$-th entry is $\tau_{\cE,i} =\cE(\ket{i}\bra{i})$),  acceleration coefficient $\gamma \in [\eta_{\rm con}(\cE),1]$, and additive error $\varepsilon>0$
     \State Choose $\lambda^{(1)}_i=\frac{1}{|X|}$ for all $i \in \{1,2,\dots,|X| \}$
     \For{$t \in \{1,2,\dots,n= \lceil \gamma \log |X|/\varepsilon \rceil\}$}
        \State $\lambda^{(t+1)}_i =  \frac{1}{Z^{(t+1)}}\lambda^{(t)}_i\expB{\frac1\gamma D\left(\tau_{\cE,i} || \cE(\rho_{\lambda^{(t)}})\right)}$, where $\rho_{\lambda^{(t)}}= \sum_k\lambda_k^{(t)} \ket{k}\bra{k}$, $\tau_{\cE,i} = \cE(\ket{i}\bra{i})$ and $Z^{(t+1)}=\sum_i \lambda^{(t)}_i\expB{\frac1\gamma D\left(\tau_{\cE,i} || \cE(\rho_{\lambda^{(t)}})\right)}$
      \EndFor
   \State Outputs: $\lambda^{(n+1)}$, $\chi^{(n)}=J_\gamma(\lambda^{(n+1)}, \lambda^{(n)}, \cE)$ with $|\chi(\cE)-\chi^{(n)}| \leq \varepsilon$
   \end{algorithmic}
\end{algorithm}

\begin{rmk}[Concavity]\label{rmk:concavity}
The function $J_\gamma(\lambda, \mu, \cE) =- \gamma  \,\Tr{\rho_\lambda \log \rho_\lambda} +\Tr{ \rho_\lambda (\gamma \log \rho_\mu +\HE(\rho_\mu))}$ is concave in $\lambda$ (for fixed $\mu$). This follows from the concavity of the von Neumann entropy and by noting that the second term is linear in $\lambda$. In contrast to the two-variable function used in the classical Blahut-Arimoto algorithm~\cite{Blahut72,Arimoto72}, the function $J_\gamma$ is not concave in $\mu$. To see this, consider the classical channel given by the stochastic matrix
\begin{align}
Q = \begin{pmatrix}
\frac{1}{2} & 1 & 0\\
\frac{1}{2} & 0 & 1\\
0 & 0 & 0\\
\end{pmatrix}\,,
\end{align}
$\lambda = (\frac{1}{2},\frac{1}{2},0), \, \mu_1 = (\frac{1}{2},\frac{1}{2},0)$, and $\mu_2 = (\frac{1}{2}, 1/4, 1/4)$. Indeed, for $a = 0.7$ we have $J_1(\lambda, a\mu_1 +(1-a)\mu_2, \cE)  < a J_1(\lambda, \mu_1, \cE)+(1-a)J_1(\lambda, \mu_2, \cE)$ showing that the function is not concave. On the other hand, for $a = 0.3$ we have $J_1(\lambda, a\mu_1 +(1-a)\mu_2, \cE)  > a J_1(\lambda, \mu_1, \cE)+(1-a)J_1(\lambda, \mu_2, \cE)$ showing that the function is not convex either. 
\end{rmk}

Arimoto~\cite{Arimoto72} showed in the classical case that if the optimizer that achieves the capacity is unique and full rank, there exists an integer $N\in \mathbb{N}$ such that for all $t>N$, the convergence rate is exponential. That is, for $t>N$, there exists $0 <\delta \leq 1$ such that
\begin{equation}
    D(\lambda^\star||\lambda^{(t)}) \leq (1-\delta)^{(t-N)}D(\lambda^\star||\lambda^{(N)})
\end{equation}
Li and Cai~\cite{li2019computing} show an analogous result for cq channels. Here we provide a different criterion that yields exponential convergence from the first iteration albeit with a stronger assumption. First, similar to \eqref{eq:contraction-coefficient}, we define the relative entropy expansion coefficient as
\begin{align}\label{eq:expansion-coefficient}
\eta_{\rm exp}(\cE)=\inf \frac{D\left(\cE(\rho)\middle\|\cE(\sigma)\right)}{D\left(\rho\middle\|\sigma\right)}\, \, \, \text{for $\sigma\gg\rho$ and $\sigma\neq\rho$}.
\end{align} 
Note that $\eta_{\rm exp}(\cE) \in [0, 1]$ by the nonnegativity of relative entropy and by the data processing inequality. Using~\eqref{eq:useful_form_bounding_cond_cq}, we see that when the coefficient $\eta_{\rm \exp}(\cE) > 0$, we get exponential convergence according to Proposition~\ref{prop:exp_convergence} with $\delta = \frac{\eta_{\rm exp}(\cE)}{\gamma}$. A consequence of the condition $\eta_{\rm \exp}(\cE) > 0$ is that $\forall \rho\neq\sigma, \ D(\cE(\rho)||\cE(\sigma)) \neq 0$ and hence $\cE(\rho)\neq\cE(\sigma)$ i.e. the map is injective.

 
\subsection{Simulation results}\label{sec:numerics-cq}

\begin{figure}[ht]
    \center
    \includegraphics[width=.45\textwidth]{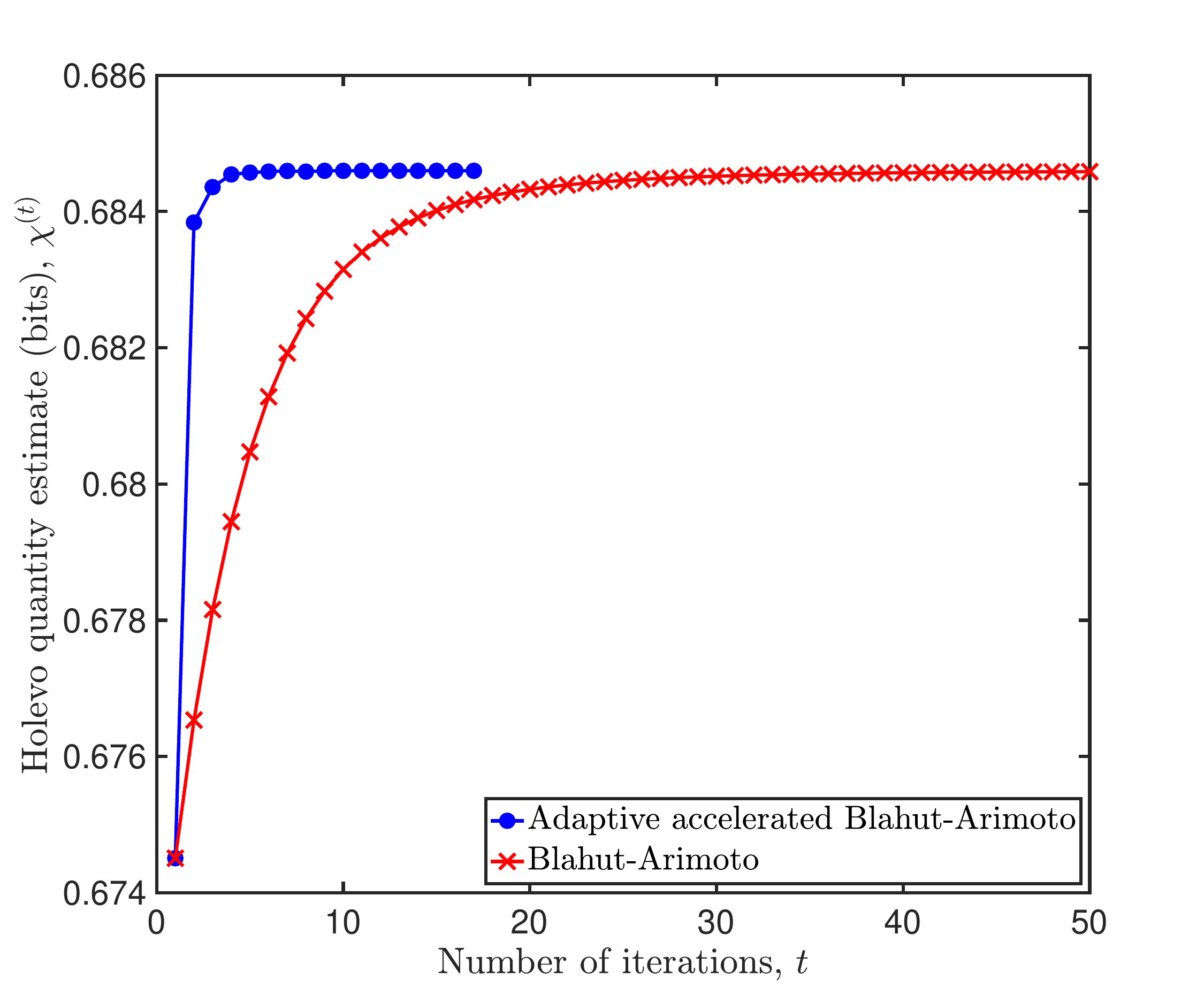}
    \includegraphics[width=.45\textwidth]{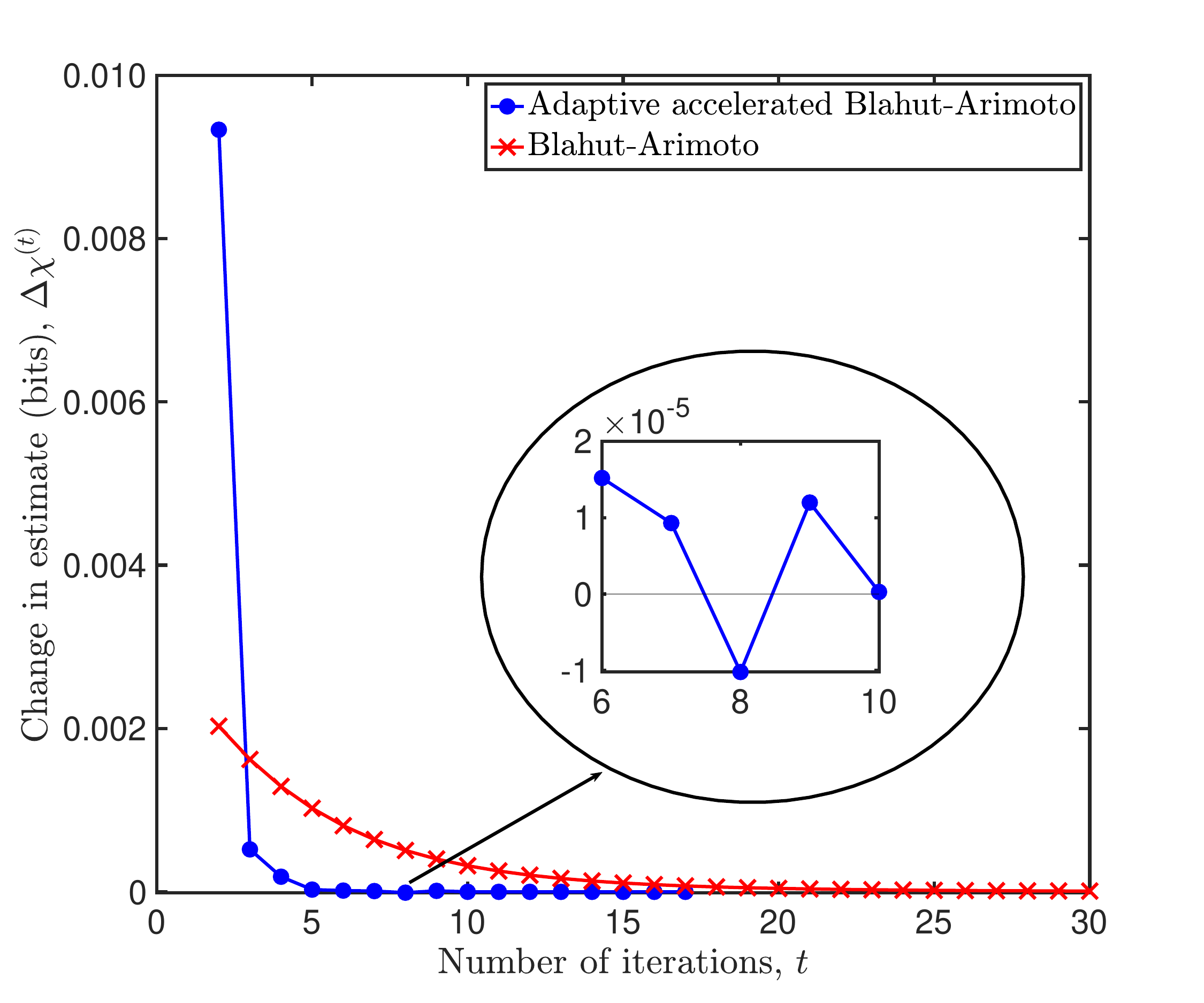}
    \caption{Convergence of the standard and adaptive accelerated Blahut-Arimoto algorithm to the Holevo quantity of a random cq channel with input alphabet of size $10$ and output dimension $16$. The left panel shows the lower bound on the capacity in each iteration step $t$ until the \textit{a posteriori} bound given in Proposition~\ref{prop:dynamic_termination} ensures that we terminate when $|\chi(\cE) - \chi^{(t)}| \leq10^{-6}$ bits. The adaptive accelerated algorithm with the acceleration parameter $\gamma^{(t)}$ from~\eqref{eq:adaptivegamma} converges in $17$ iterations, while the standard version converges after $248$ iterations (displayed up to the $50^{\rm th}$ iteration step). The right panel shows the improvement obtained in the capacity estimate in each iteration (displayed up to the $30^{\rm th}$ iteration step). The zoomed inset in the panel shows that the adaptive accelerated algorithm does not necessarily have a monotonically increasing capacity estimate.} 
    \label{fig:CQ_random_states}
\end{figure}

We numerically compute the classical capacity of a cq channel with the Holevo quantity $\chi(\mathcal{E})$ as given in~\eqref{cq_capacity}. We choose the ensemble of output density operators randomly using the \texttt{RandomDensityMatrix} package from QETLAB~\cite{qetlab}. Figure~\ref{fig:CQ_random_states} illustrates the results for a channel with input alphabet of size $10$ and output dimension $16$. The initial guess $\rho_\lambda^{(1)}$ is chosen to be the maximally mixed state. For several simple choices of channels, the uniform distribution is optimal and our algorithm requires only one iteration to compute the capacity (this is the reason for choosing a random channel for the illustration here). By Proposition~\ref{thm:convergence}, the capacity estimate $C(n)$ is $\varepsilon$-close to $C^\star$ after $\lceil\frac{\gamma\log|X|}{\varepsilon}\rceil$ iterations. Setting $\varepsilon=10^{-6}$ bits,  it would require around $10^6$ iterations for a provable convergence. To reduce the number of iterations, we use the termination criterion given in Proposition~\ref{prop:dynamic_termination}. Hence, if we observe that the maximal eigenvalue of $\HE(\rho^{(t)})$ satisfies $\alpha^{(t)}_{\max} - C(t)\leq\varepsilon$ in the $t$-th iteration step, we also have that $C^\star - C(t)\leq\varepsilon$. With this termination condition, the standard algorithm converges in $248$ iterations while the adaptive accelerated algorithm converges in $17$ iterations. 

We may also compare our numerics to the one given by Fawzi and Fawzi~\cite{Fawzi18}. Their approach to compute the capacity of this particular cq channel requires $2.21$ seconds using CVX and the Mosek solver~\cite{Lofberg2004}. Their method does not allow us to bound the error between the capacity estimate and the true capacity but the solver tolerance achieved is $1.3\times 10^{-6}$ bits. For the same channel, our standard Blahut-Arimoto algorithm takes $0.98$ seconds to achieve an \textit{a posteriori} error of $10^{-6}$ bits and the adaptive accelerated algorithm takes $0.09$ seconds.\footnote{All run-times correspond to a machine with a 2.3 GHz Intel Core i5 processor with 8GB of RAM.}


\section{Thermodynamic capacity of quantum channels}\label{sec:thermo}

\subsection{Definitions}

The thermodynamic capacity quantifies the information-theoretic power of quantum channels in the presence of physical restrictions imposed by thermodynamics~\cite{faist_thermodynamic_2018} (see also~\cite{Navascues15}). For a quantum channel $\cE_{A \rightarrow B}$, relative to operators $\Gamma_A,\Gamma_B > 0$ it can be written as
\begin{align}
T_\Gamma (\cE) &= \max_{\rho_A }  \underbrace{D(\cE(\rho_A)|| \Gamma_B) -D(\rho_A|| \Gamma_A)}_{=K(\rho_A, \cE)} \label{eq:K_rho_defn}\\
&=\max_{\rho_A } S(\rho_A) - S(\cE(\rho_A)) + \Tr{\rho_A \left\{ \log\Gamma_A - \cE^\dagger\left(\log\Gamma_B\right)\right\}} \, .
\end{align}
The $\Gamma$ operators are thereby typically given by Gibbs states generated by Hamiltonians  (i.e., some Hermitian operators) $H_A$ and $H_B$ on the input and output systems, respectively. That is, we have the choice $\Gamma_A=\expB{-\beta H_A}$ for a fixed inverse temperature parameter $\beta\geq0$ and similarly for the output system $B$. However, for the sake of our algorithm we leave $\Gamma_A,\Gamma_B > 0$ in general form and refer to~\cite{faist_thermodynamic_2018} for a discussion of the thermodynamic capacity and its properties. An interesting special case is $\Gamma_A=1_A$ and $\Gamma_B=1_B$, for which we get $T_1 (\cE)=-G(\cE)$, where the minimal entropy gain of the quantum channel is given by~\cite{Alicki04,Holevo11,Holevo2011}
\begin{align}
G(\cE)=\min_{\rho_A } S(\cE(\rho_A))-S(\rho_A)\,.
\end{align}
We refer to~\cite{Buscemi16} for a discussion of the minimal entropy gain for finite-dimensional quantum channels. The thermodynamic capacity is also related to the completely bounded minimal conditional entropy~\cite{Devetak06}
\begin{align}
S_{\text{CB},\min}(\cE)=\min_{\rho_A } S((\cE\otimes\mathcal{I})(\rho_{AR}))-S(\rho_A)
\end{align}
with $\rho_{AR}$ a purification of $\rho_A$. Namely, we have $S_{\text{CB},\min}(\cE)=-T_1(\cE^c)$ with $\cE^c_{A\to E}$ the complementary channel of $\cE_{A\to B}$. The completely bounded minimal conditional entropy plays a role in entanglement theory as shown in~\cite[Section 5]{Devetak06}.


\subsection{Blahut-Arimoto algorithm}

The thermodynamic capacity can be approximated using a Blahut-Arimoto algorithm of the structure given in Algorithm~\ref{algo:general_AB}. We define the following two variable extension of $K(\rho_A, \cE)$ for $\sigma \gg \rho$ and $\gamma \in [ 1 -\eta_{\rm exp}(\cE),1]$, where $K(\rho_A, \cE)$ is defined in \eqref{eq:K_rho_defn} and $\eta_{\rm exp}(\cE)$ is defined in \eqref{eq:expansion-coefficient}. 
\begin{align} \label{eq:J_gamma_thermo}
J_\gamma(\rho, \sigma, \cE) & = K(\rho, \cE) +(1-\gamma) D(\rho||\sigma) -D(\cE(\rho)||\cE(\sigma)) \, ,
\end{align}
where we omitted the system indices for simplicity. Note that $\gamma = 1$ corresponds to the standard algorithm and if it is known that $\eta_{\rm exp}(\cE) > 0$, then one may accelerate the algorithm. With a short calculation, one can bring $J_\gamma$ into the form
\begin{align}
J_\gamma(\rho, \sigma, \cE)
& =  - \gamma \, \Tr{\rho \log \rho} +  \Tr{\rho \left\{\gamma \, \log \sigma +  \HE(\sigma) \right\}} \textnormal{ with } \label{eq:J_gamma_thermo_form}\\
\HE(\sigma) &= -\logs \sigma + \cE^\dagger(\logs \cE(\sigma)) +\log\Gamma_A - \cE^\dagger(\log \Gamma_B) \, . \label{def:H_thermo}
\end{align}
Note that $\Tr{\rho\HE(\sigma) }$ is continuous for $\sigma \gg \rho$. Moreover, another short calculation leads to 
\begin{align}
&\Tr{\rho\left\{\HE(\sigma) - \HE(\rho)\right\}}
=D(\rho||\sigma)-D(\cE(\rho)||\cE(\sigma)) \textnormal{ and hence } \label{eq:useful_form_bounding_cond_thermo}\\
&0 \leq \Tr{\rho\left\{\HE(\sigma) - \HE(\rho)\right\}} \leq \gamma D(\rho||\sigma) \, ,  \label{eq:requirement_thermo}
\end{align}
where we used the data processing inequality for the quantum relative entropy (together with $\gamma \geq 1 - \eta_{\rm exp}(\cE)$). 
From Lemma~\ref{lem:update_rules} whose requirement is satisfied by~\eqref{eq:requirement_thermo}, we find a double optimization form of the thermodynamic capacity $T_\Gamma (\cE)$
\begin{align}
\max_{\rho,\sigma \textnormal{ with } \sigma \gg \rho} J_\gamma(\rho, \sigma, \cE)&=\max_{\rho} J_\gamma(\rho, \rho, \cE)=\max_{\rho} K(\rho, \cE) =T_\Gamma (\cE)\,.
\end{align}

Performing the two maximizations in $ \max_{\rho, \sigma} J_\gamma(\rho, \sigma, \cE)$ iteratively, leads to the following algorithm  (see Lemma~\ref{lem:update_rules} for the form of the update rules), which provably converges  to the thermodynamic capacity by Theorem~\ref{thm:convergence}. A detailed analysis of the time complexity can be found in Appendix~\ref{app:time_compl_thermo}. 


\begin{algorithm}[H] 
  \caption{Accelerated Blahut-Arimoto type algorithm for the thermodynamic capacity}
\label{algo:AB_thermo_acc}
   \begin{algorithmic}[1]
    \State Inputs: Quantum channel $\cE_{A\rightarrow B}$ and its adjoint $\cE_{B\rightarrow A}^\dagger$ (both given as lookup tables whose $(i,j)$-th entry is given by $\cE(\ket{i}\bra{j})$ or $\cE^\dagger(\ket{i}\bra{j})$, respectively), acceleration coefficient $\gamma \in [1 -\eta_{\rm exp}(\cE),1]$  and additive error $\varepsilon>0$
    \State Choose $\rho^{(1)}=\frac{1_A}{|A|}$ 
     \For{$t \in \{1,2,\dots, n=\lceil\gamma \log |A|/\varepsilon \rceil\}$}
        \State $\rho^{(t+1)} = \frac{1}{Z^{(t+1)}} \expB{\log\rho^{(t)} +\frac1\gamma \HE(\rho^{(t)})} \, ,$ where 
$\HE(\sigma) = -\log\sigma + \cE^\dagger(\logs \cE(\sigma)) +\log\Gamma_A - \cE^\dagger\left(\log\Gamma_B\right)$ and $Z^{(t+1)}=\Tr{\expB{\log\rho^{(t)} +\frac1\gamma \HE(\rho^{(t)}))}}$ normalizes the state.
      \EndFor
   \State Outputs: $\rho^{(n+1)}$, $T_{\Gamma}^{(n)}=J_\gamma(\rho^{(n+1)},\rho^{(n)},\cE)$ with $|T_\Gamma(\cE)-T_{\Gamma}^{(n)}| \leq \varepsilon$
   \end{algorithmic}
\end{algorithm}


When the coefficient $\eta_{\rm con}(\cE) < 1$ (from \eqref{eq:contraction-coefficient}), we get exponential convergence. To see this, note
\begin{align}
    1-\eta_{\rm con}(\cE) &= 1 - \sup_{\rho\neq\sigma} \frac{D(\cE(\rho)||\cE(\sigma))}{D(\rho||\sigma)} \\
    &= \inf_{\rho\neq\sigma} \frac{D(\rho||\sigma)-D(\cE(\rho)||\cE(\sigma))}{D(\rho||\sigma)} \, .
\end{align}
Using~\eqref{eq:useful_form_bounding_cond_thermo}, we see that we have exponential convergence according to Proposition~\ref{prop:exp_convergence} with $\delta = \frac{1 - \eta_{\rm con}(\cE)}{\gamma}$ when the coefficient $\eta_{\rm con}(\cE) < 1$.

\subsection{Simulation results}\label{sec:numerics-thermo}

\begin{figure}[ht]
    \centering
    \includegraphics[width=.45\textwidth]{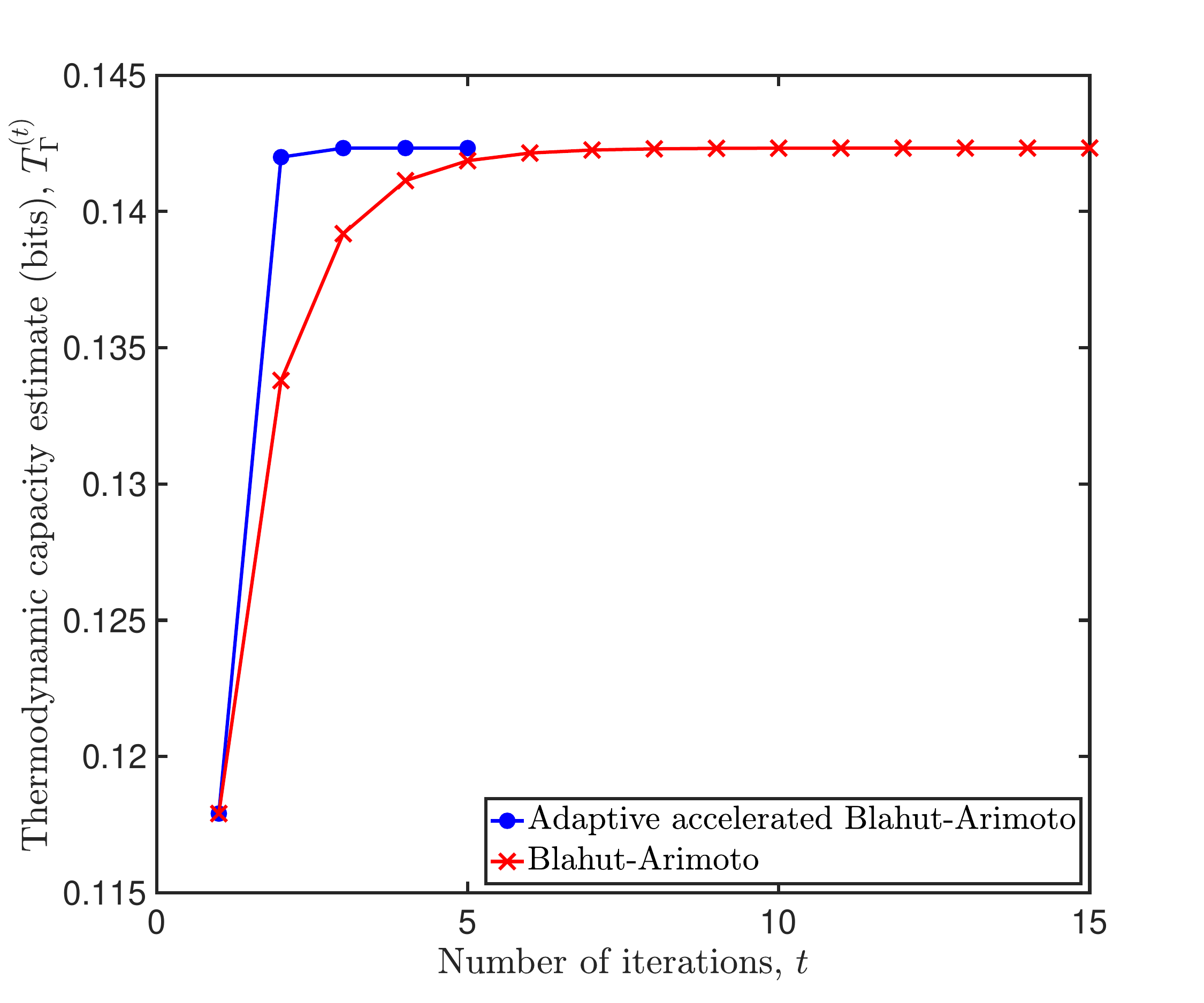}
    \includegraphics[width=.45\textwidth]{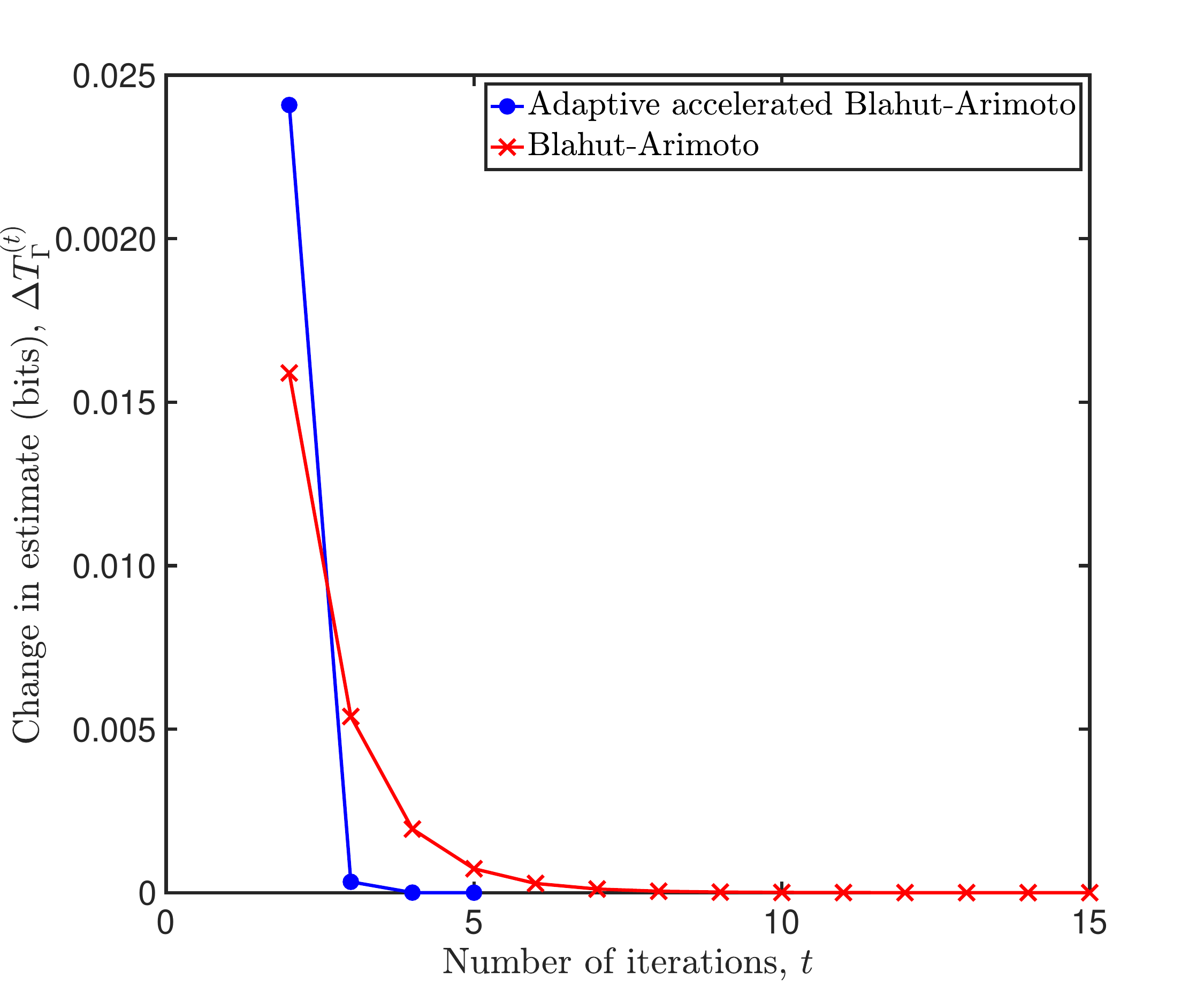}
    \caption{Convergence of the Blahut-Arimoto algorithm to the thermodynamic capacity of the amplitude damping channel $\mathcal{E}^{AD}_{0.3}$ (given in~\ref{eq:amplitude_damping_channel}) in the standard and adaptive accelerated case with the adaptive acceleration parameter $\gamma^{(t)}$ in the $t$-th iteration step (where $\gamma^{(t)}$ is given in~\eqref{eq:adaptivegamma}). The left panel shows the lower bound on the thermodynamic capacity in each iteration step $t$ until the \textit{a posteriori} bound given in Proposition~\ref{prop:dynamic_termination} ensures that we terminate when $|T_{\Gamma}(\cE) - T_{\Gamma}^{(t)}| \leq10^{-6}$ bits. The adaptive accelerated Blahut-Arimoto algorithm converges after $5$ iterations, while the standard algorithm converges after $25$ iterations (displayed up to 15 iterations). The right panel shows the improvement obtained in the capacity estimate with each iteration.}
    \label{fig:TC_amplitude_damping}
\end{figure}

For unital channels, the thermodynamic capacity is zero and the maximizer is the maximally mixed state, which we have chosen as our initial guess $\rho_\lambda^{(1)}$. Here, we consider the non-unital qubit amplitude damping channel which we recall from \eqref{eq:amplitude_damping_channel_intro}
\begin{align}\label{eq:amplitude_damping_channel}
\text{$\mathcal{E}^{AD}_{p}(\rho) =A_{0} \rho A_{0}^{\dagger}+A_{1} \rho A_{1}^{\dagger}$ with $A_{0} = | 0 \rangle\langle 0|+\sqrt{1-p}| 1\rangle\langle 1\vert, A_{1} = \sqrt{p}\vert 0\rangle\langle 1\vert$ for $p \in[0,1]$.}
\end{align}
As in Section~\ref{sec:numerics-cq}, we choose an additive error $\varepsilon=10^{-6}$ bits. We use the termination criterion given in Proposition~\ref{prop:dynamic_termination}, which significantly reduces the required number of iterations of the algorithm. The results are shown in Figure~\ref{fig:TC_amplitude_damping}.
The standard and the  accelerated algorithm converge in $25$ and $5$ iterations, respectively.


\section{Coherent information of less noisy channels}\label{sec:ci}

\subsection{Definitions}
For a quantum channel $\cE$, the coherent information of the channel, $I_{\rm{coh}}(\cE)$, is defined as the maximum of the coherent information $I_{\rm{coh}}(\rho, \cE) = S(\cE(\rho)) - S(\cE_c(\rho))$, where the maximization is taken over input states $\rho$. That is, we have
\begin{align} \label{eq:defintion_ci}
I_{\rm{coh}}(\cE) = \max_{\rho} \underbrace{S(\cE(\rho)) - S(\cE_c(\rho))}_{=I_{\rm{coh}}(\rho, \cE)} \, .
\end{align}
A channel $\cE$ is called \textit{less noisy} if the private capacity of its complementary channel $\cE_c$ is zero~\cite{watanabe2012private}. In~\cite[Proposition 4]{watanabe2012private}, it is shown that a channel is less noisy if and only if it is less divergence contracting, which implies that for all density operators $\rho$ and $\sigma$, we have
\begin{align}\label{eq:less_noisy_condition}
D(\cE(\rho)||\cE(\sigma)) \geq D(\cE_c(\rho)||\cE_c(\sigma))\,.
\end{align}
The quantum capacity $Q(\cE)$ of the channel $\cE$ is the regularized coherent information of a channel~\cite{devetak2005private} and corresponds to the rate at which Alice can transmit quantum information to Bob. In~\cite{watanabe2012private} it is shown that the coherent information of a channel is additive if the channel is less noisy. Hence, for less noisy channels, we have $Q(\cE) = I_{\rm{coh}}(\cE)$. We shall only consider less noisy channels henceforth and we therefore drop the distinction between the regularized and single letter formulas. 


\subsection{Blahut-Arimoto algorithm}

The coherent information of a quantum channel can be approximated using a Blahut-Arimoto algorithm of the structure given in Algorithm~\ref{algo:general_AB}. We may also accelerate the algorithm and to do so, we introduce the following coefficient.
\begin{align}\label{eq:noisy_contraction_coeff}
    \zeta_{\rm con}(\cE) &= \sup_{\rho\neq\sigma}\frac{D(\cE(\rho)||\cE(\sigma)) - D(\cE_c(\rho)||\cE_c(\sigma))}{D(\rho||\sigma)}\,.
\end{align}

We define the following two variable extension of $I_{\rm{coh}}(\rho, \cE)$ for $\sigma \gg \rho$ and $\gamma \in [\zeta_{\rm con}, 1]$.
\begin{align} \label{eq:J_gamma_ci}
J_\gamma(\rho, \sigma, \cE) & = I_{\rm{coh}}(\rho, \cE) + D(\cE(\rho)||\cE(\sigma)) - D(\cE_c(\rho)||\cE_c(\sigma)) - \gamma D(\rho||\sigma)\, .
\end{align}

With a short calculation, one can bring $J_\gamma$ into the following form
\begin{align}
J_\gamma(\rho, \sigma, \cE) 
& =  - \gamma \, \Tr{\rho \log \rho} +  \Tr{\rho \left\{\gamma \, \log \sigma +  \HE(\sigma) \right\}}\label{eq:J_gamma_ci_form}\textnormal{ with}\\
\HE(\sigma) &= \cE_c^\dagger\log\cE_c(\sigma) - \cE^\dagger\log\cE(\sigma)\, . \label{def:H_ci}
\end{align}
Note that $\gamma = 1$ corresponds to the standard algorithm and if $\zeta_{\rm con}(\cE) < 1$, then one may accelerate the algorithm. We also note that $\Tr{\rho\HE(\sigma) }$ is continuous for $\sigma \gg \rho$. Moreover, another short calculation leads to 
\begin{align}
&\Tr{\rho\left\{\HE(\sigma) - \HE(\rho)\right\}}
= D(\cE(\rho)||\cE(\sigma)) - D(\cE_c(\rho)||\cE_c(\sigma))\textnormal{ and hence }\label{eq:useful_form_bounding_cond_ci} \\
&0 \leq \Tr{\rho\left\{\HE(\sigma) - \HE(\rho)\right\}} \leq \gamma D(\rho||\sigma) \, , \label{eq:requirement_ci}
\end{align}
where we have used \eqref{eq:less_noisy_condition} along with $\gamma\geq \zeta_{\rm con}(\cE)$. From Lemma~\ref{lem:update_rules}, whose requirement is satisfied by~\eqref{eq:requirement_ci}, we find a double optimization form of the coherent information $I_{\rm{coh}}(\cE)$
\begin{align}
\max_{\rho,\sigma \textnormal{ with } \sigma \gg \rho} J_\gamma(\rho, \sigma, \cE)&=\max_{\rho} J_\gamma(\rho, \rho, \cE)=\max_{\rho} I_{\rm{coh}}(\rho, \cE) = I_{\rm{coh}}(\cE)\,.
\end{align}

Performing the two maximizations in $ \max_{\rho, \sigma} J_\gamma(\rho, \sigma, \cE)$ iteratively, leads to the following algorithm (see Lemma~\ref{lem:update_rules} for the form of the update rules), which provably converges  to  the coherent information by Theorem~\ref{thm:convergence}.   A detailed analysis of the time complexity is given in Appendix~\ref{app:time_compl_ci}. 


\begin{algorithm}[H] 
  \caption{Accelerated Blahut-Arimoto type algorithm for the coherent information}
\label{algo:AB_ci_acc}
   \begin{algorithmic}[1]
    \State Inputs: Quantum channel $\cE_{A\rightarrow B}$, its complementary channel $\cE_{c}$ and the respective adjoint channels $\cE_{B\rightarrow A}^\dagger$ and $\cE_c^\dagger$ (all given as lookup tables whose $(i,j)$-th entry is given by the action of the channel on $\ket{i}\bra{j}$), acceleration coefficient $\gamma\in [\zeta_{\rm con}, 1]$  and additive error $\varepsilon>0$
    \State Choose $\rho^{(1)}=\frac{1_A}{|A|}$ 
     \For{$t \in \{1,2,\dots, n=\lceil\gamma \log |A|/\varepsilon \rceil\}$}
        \State $\rho^{(t+1)} = \frac{1}{Z^{(t+1)}} \expB{\log\rho^{(t)} +\frac1\gamma \HE(\rho^{(t)})} \, ,$ where 
$\HE(\sigma) = \cE_c^\dagger\log\cE_c(\sigma) - \cE^\dagger\log\cE(\sigma)$ and 

$Z^{(t+1)}=\Tr{\expB{\log\rho^{(t)} +\frac1\gamma \HE(\rho^{(t)}))}}$ normalizes the state.
      \EndFor
   \State Outputs: $\rho^{(n+1)}$, $I_{\rm{coh}}^{(n)}=J_\gamma(\rho^{(n+1)},\rho^{(n)},\cE)$ with $|I_{\rm{coh}}(\cE)-I_{\rm{coh}}^{(n)}| \leq \varepsilon$
   \end{algorithmic}
\end{algorithm}

We also have exponential convergence if the channel is strictly less noisy i.e. $\inf D(\cE(\rho)||\cE(\sigma)) - D(\cE_c(\rho)||\cE_c(\sigma)) > 0$ for all $\sigma \gg \rho$ and $\sigma\neq\rho$. To see this, we define
\begin{align}
\label{eq:noisy_expansion_coefficient}
    \zeta_{\rm exp}(\cE) &= \inf \frac{D(\cE(\rho)||\cE(\sigma)) - D(\cE_c(\rho)||\cE_c(\sigma))}{D(\rho||\sigma)} \, \, \, \text{for $\sigma\gg\rho$ and $\sigma\neq\rho$.}
\end{align}
In general, $\zeta_{\rm exp}(\cE) \in [-1, 1]$ but for less noisy channels, we have that $\zeta_{\rm exp}(\cE) \in [0, 1]$ due to~\eqref{eq:less_noisy_condition}. Comparing it to~\eqref{eq:noisy_contraction_coeff}, we also note that $\zeta_{\rm exp}(\cE) = -\zeta_{\rm con}(\cE_c)$. Using~\eqref{eq:useful_form_bounding_cond_ci}, we see that when the coefficient $\zeta_{\rm exp}(\cE) > 0$, we get exponential convergence according to Proposition~\ref{prop:exp_convergence} with $\delta = \frac{\zeta_{\rm exp}(\cE)}{\gamma}$. Indeed, if the channel is strictly less noisy, the condition $\zeta_{\rm exp}(\cE)> 0$ follows. 


\subsection{Simulation results}\label{sec:numerics-ci}

\begin{figure}[ht]
    \centering
    \includegraphics[width=.45\textwidth]{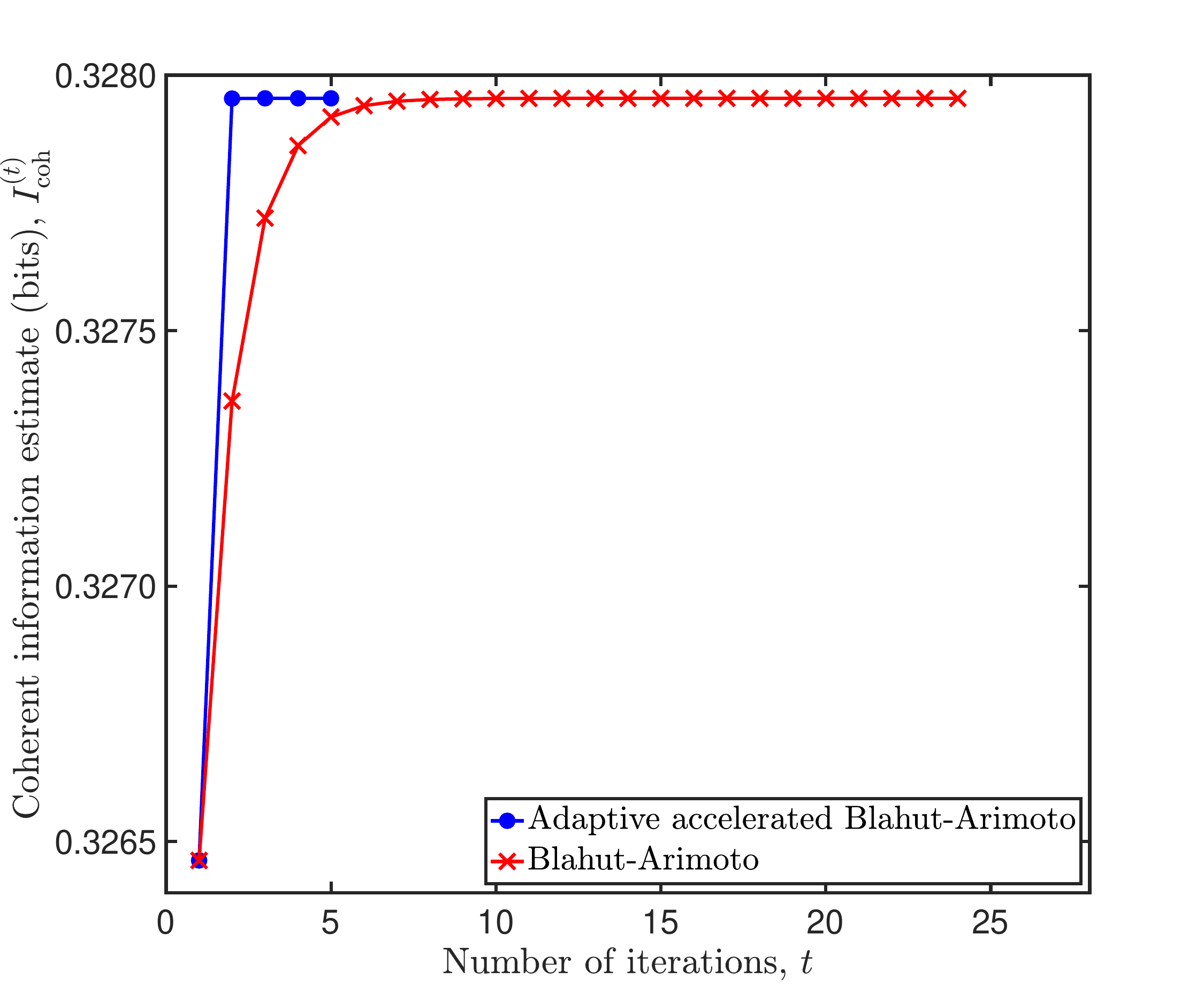}
    \includegraphics[width=.45\textwidth]{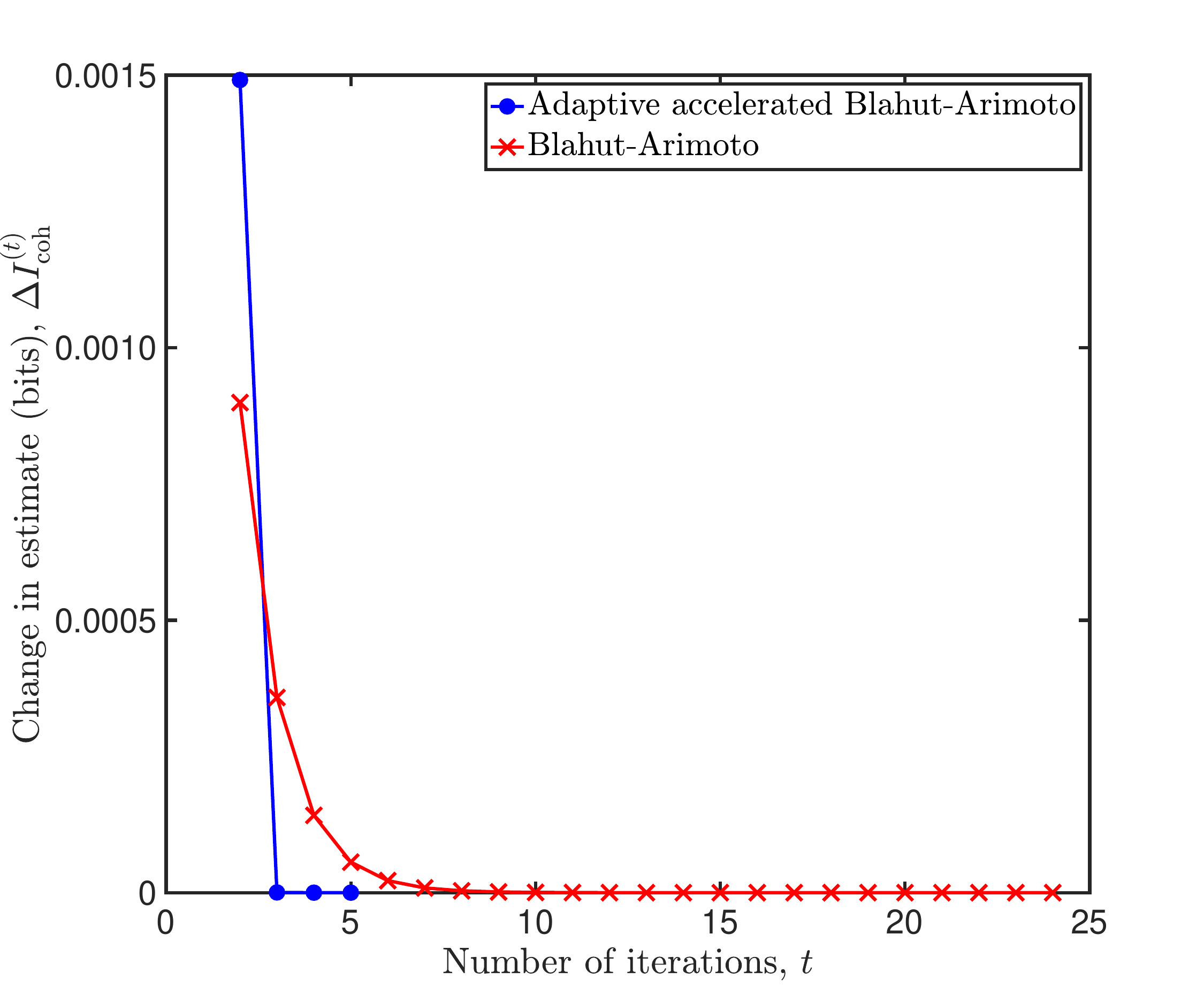}
    \caption{Convergence of the Blahut-Arimoto algorithm to the coherent information of the amplitude damping channel $\mathcal{E}^{AD}_{0.3}$ given in~\eqref{eq:amplitude_damping_channel} in the standard and adaptive accelerated case with acceleration parameter $\gamma^{(t)}$ in the $t$-th iteration step as defined in \eqref{eq:adaptivegamma}. The left panel shows the lower bound on the coherent information in each iteration step $t$ until the \textit{a posteriori} bound given in Proposition~\ref{prop:dynamic_termination} ensures that we terminate when $|I_{\rm coh}(\cE) - I_{\rm coh}^{(t)}| \leq10^{-6}$ bits. The adaptive accelerated Blahut-Arimoto algorithm converges after $5$ iterations, while the standard algorithm converges after $24$ iterations. The right panel shows the improvement obtained in the coherent information estimate with each iteration.}
    \label{fig:CI_amplitude_damping}
\end{figure}
Once again we consider the qubit amplitude damping channel whose form is given in~\eqref{eq:amplitude_damping_channel}. We choose an additive error $\varepsilon=10^{-6}$ bits and use the termination criterion given in Proposition~\ref{prop:dynamic_termination} which significantly reduces the number of iterations of the algorithm required to achieve an estimate of the capacity with additive error smaller than $\varepsilon$. The results are shown in Figure~\ref{fig:CI_amplitude_damping}. The standard algorithm and the  adaptive accelerated algorithm converge in $24$  and $5$ iterations, respectively.


\section{Mutual information of quantum channels}\label{sec:mi}

\subsection{Definitions}

The mutual information $I(\cE)$ of a quantum channel $\cE$ is defined as the maximum over input states $\rho$ of the mutual information $I(\rho, \cE) = S(\rho) + S(\cE(\rho)) - S(\cE_c(\rho))$, i.e.,
\begin{align} \label{eq:definition_mi}
I(\cE) = \max_{\rho} \underbrace{S(\rho) + S(\cE(\rho)) - S(\cE_c(\rho))}_{=I(\rho, \cE)} \, ,
\end{align}

The entanglement assisted classical capacity $C(\cE)$ of the channel $\cE$ is the regularized mutual information of a channel~\cite{Bennett02}. It is also known that the mutual information of a channel is additive~\cite{wildebook17} and hence we have $C(\cE) = I(\cE)$. 


\subsection{Blahut-Arimoto algorithm}

The mutual information of a quantum channel can be approximated using a Blahut-Arimoto algorithm of the structure given in Algorithm~\ref{algo:general_AB}. First, recall the definition of $\zeta_{\rm con}$ given in~\eqref{eq:noisy_contraction_coeff}. For any $\gamma \in [1 + \zeta_{\rm con}(\cE),2]$,
we define the following two variable extension of $I(\rho, \cE)$ for $\sigma \gg \rho$ 
\begin{align} \label{eq:J_gamma_mi}
J_\gamma(\rho, \sigma, \cE) & = I(\rho, \cE) + D(\cE(\rho)||\cE(\sigma)) - D(\cE_c(\rho)||\cE_c(\sigma)) + (1-\gamma) D(\rho||\sigma)\, .
\end{align}
Note that $\gamma = 2$ corresponds to the standard algorithm. If it is known that $\zeta_{\rm con}(\cE) < 1$, then one may accelerate the algorithm. With a short calculation, one can bring $J_\gamma$ into the form
\begin{align}
J_\gamma(\rho, \sigma, \cE) 
& =  - \gamma \, \Tr{\rho \log \rho} +  \Tr{\rho \left\{\gamma \, \log \sigma +  \HE(\sigma) \right\}} \textnormal{ with }\label{eq:J_gamma_mi_form} \\
\HE(\sigma) &= -\log(\sigma) + \cE_c^\dagger\log\cE_c(\sigma) - \cE^\dagger\log\cE(\sigma)\, . \label{def:H_mi}
\end{align}
Note that $\Tr{\rho\HE(\sigma) }$ is continuous for $\sigma \gg \rho$. Moreover, a short calculation leads to 
\begin{align}
\label{eq:useful_form_bounding_cond_mi}
&\Tr{\rho\left\{\HE(\sigma) - \HE(\rho)\right\}}
= D(\rho||\sigma) + D(\cE(\rho)||\cE(\sigma)) - D(\cE_c(\rho)||\cE_c(\sigma)) \, \textnormal{ and hence }\\
&0 \leq \Tr{\rho\left\{\HE(\sigma) - \HE(\rho)\right\}} \leq \gamma D(\rho||\sigma) \, ,  \label{eq:requirement_mi}
\end{align}

where we have used the data processing inequality and the nonnegativity of relative entropy along with $\gamma\geq 1+ \zeta_{\rm con}(\cE)$. From Lemma~\ref{lem:update_rules}, whose requirement is satisfied by~\eqref{eq:requirement_mi}, we find a double optimization form of the coherent information $I(\cE)$
\begin{align}
\max_{\rho,\sigma \textnormal{ with } \sigma \gg \rho} J_\gamma(\rho, \sigma, \cE)&=\max_{\rho} J_\gamma(\rho, \rho, \cE)=\max_{\rho} I(\rho, \cE)  = I(\cE)\,.
\end{align}

Performing the two maximizations in $ \max_{\rho, \sigma} J_\gamma(\rho, \sigma, \cE)$ iteratively, leads to the following algorithm (see Lemma~\ref{lem:update_rules} for the form of the update rules), which provably converges  to the mutual information of the channel $\cE$ by Theorem~\ref{thm:convergence}.  A detailed analysis of the time complexity is given in Appendix~\ref{app:time_compl_mi}.


\begin{algorithm}[H] 
  \caption{Accelerated Blahut-Arimoto type algorithm for the mutual information}
\label{algo:AB_mi_acc}
   \begin{algorithmic}[1]
    \State Inputs: Quantum channel $\cE_{A\rightarrow B}$, its complementary channel $\cE_{c}$ and the respective adjoint channels $\cE_{B\rightarrow A}^\dagger$ and $\cE_c^\dagger$ (all given as lookup tables whose $(i,j)$-th entry is given by the action of the channel on $\ket{i}\bra{j}$), acceleration coefficient $\gamma\in[1 +\zeta_{\rm con}(\cE), 2]$  and additive error $\varepsilon>0$
    \State Choose $\rho^{(1)}=\frac{1_A}{|A|}$ 
     \For{$t \in \{1,2,\dots, n=\lceil\gamma \log |A|/\varepsilon \rceil\}$}
        \State $\rho^{(t+1)} = \frac{1}{Z^{(t+1)}} \expB{\log\rho^{(t)} +\frac1\gamma \HE(\rho^{(t)})} \, ,$ where 
$\HE(\sigma) = -\log(\rho) + \cE_c^\dagger\log\cE_c(\sigma) - \cE^\dagger\log\cE(\sigma)$ and 

$Z^{(t+1)}=\Tr{\expB{\log\rho^{(t)} +\frac1\gamma \HE(\rho^{(t)}))}}$ normalizes the state.
      \EndFor
   \State Outputs: $\rho^{(n+1)}$, $I^{(n)}=J_\gamma(\rho^{(n+1)},\rho^{(n)},\cE)$ with $|I(\cE)-I^{(n)}| \leq \varepsilon$
   \end{algorithmic}
\end{algorithm}

We also have exponential convergence under certain conditions. Recall the definition $\zeta_{\rm exp}(\cE)$ from~\eqref{eq:noisy_expansion_coefficient}. Using~\eqref{eq:useful_form_bounding_cond_mi}, we see that when the coefficient $\zeta_{\rm exp}(\cE) > -1$ (or equivalently, $\zeta_{\rm con}(\cE_c) < 1$), we get exponential convergence according to Proposition~\ref{prop:exp_convergence} with $\delta =\frac{1 + \zeta_{\rm exp}(\cE)}{\gamma}$. 


\subsection{Simulation results}\label{sec:numerics-mi}

\begin{figure}[ht]
    \centering
    \includegraphics[width=.45\textwidth]{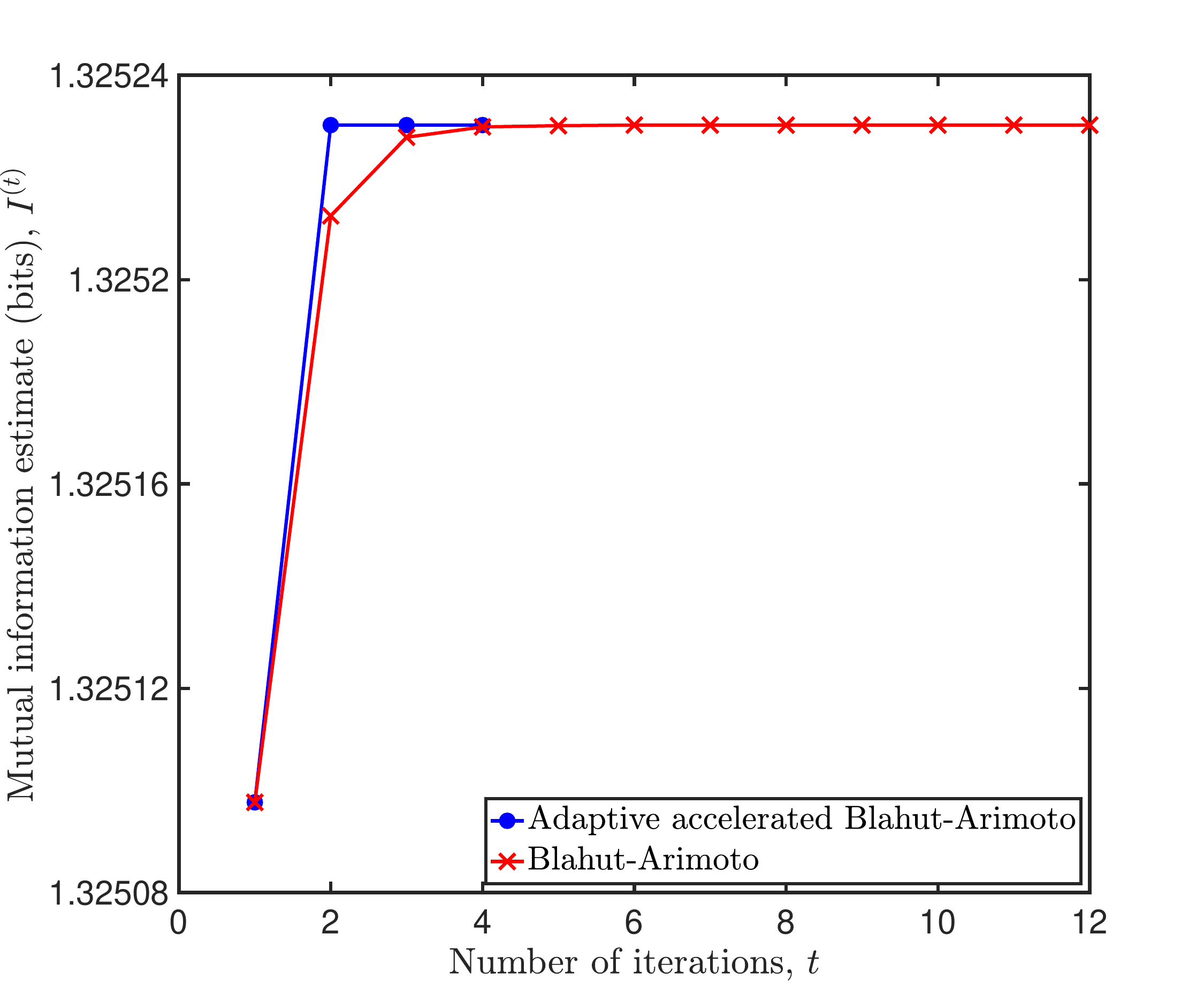}
    \includegraphics[width=.45\textwidth]{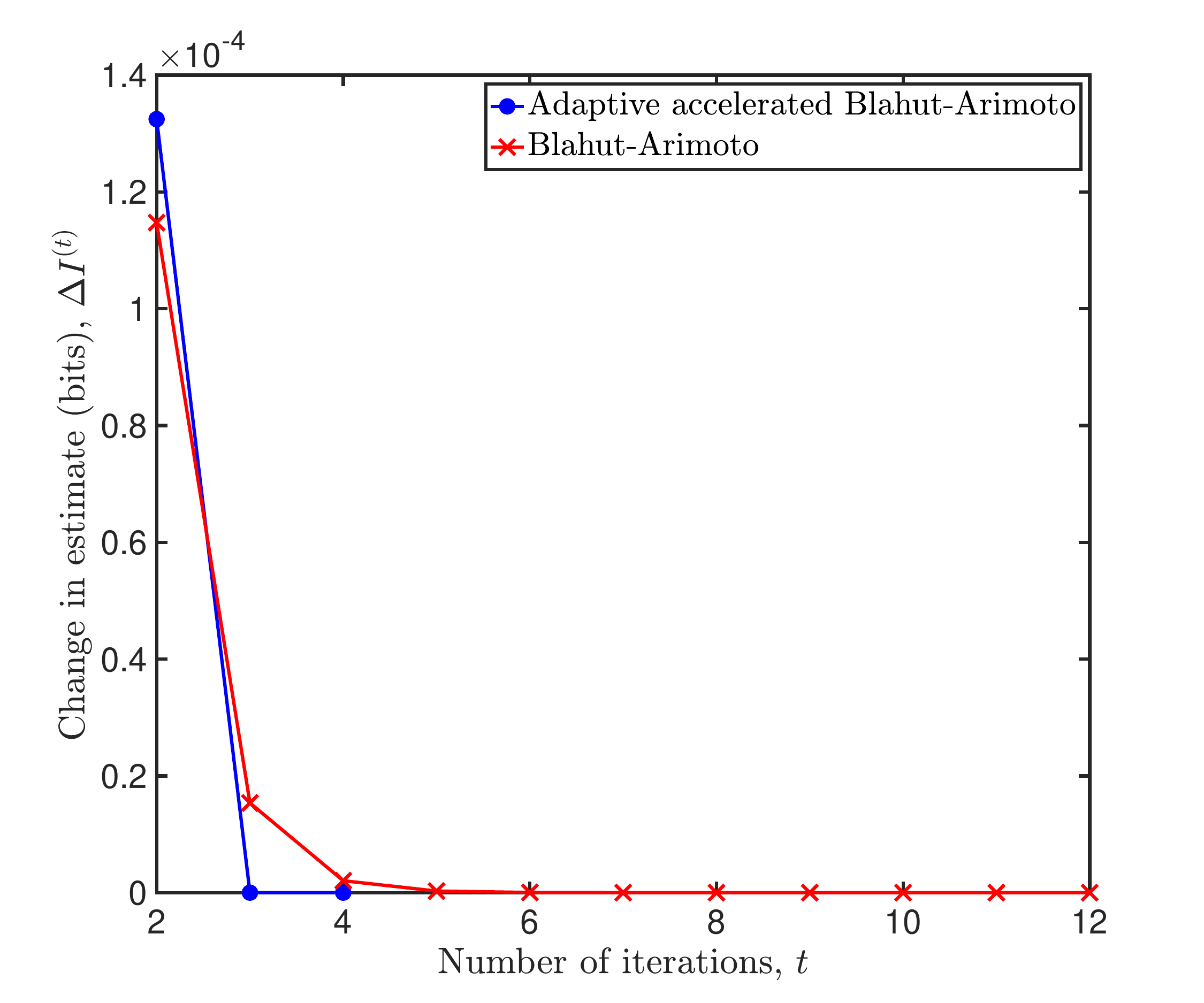}
    \caption{Convergence of the Blahut-Arimoto algorithm to the mutual information of the amplitude damping channel $\mathcal{E}^{AD}_{0.3}$ given in~\eqref{eq:amplitude_damping_channel} in the standard and adaptive accelerated case with acceleration parameter $\gamma^{(t)}$ as defined in \eqref{eq:adaptivegamma}. The left panel shows the lower bound on the mutual information in each iteration step $t$ until the \textit{a posteriori} bound given in Proposition~\ref{prop:dynamic_termination} ensures that we terminate when $|I(\cE) - I^{(t)}| \leq10^{-6}$ bits. The adaptive accelerated Blahut-Arimoto algorithm converges after $4$ iterations, while the standard algorithm converges after $12$ iterations. The right panel shows the improvement obtained in the mutual information estimate with each iteration.}
    \label{fig:MI_amplitude_damping}
\end{figure}

Once again we consider the qubit amplitude damping channel which has the form as in (\ref{eq:amplitude_damping_channel}). We choose an additive error $\varepsilon=10^{-6}$ bits and use the termination criterion given in Proposition~\ref{prop:dynamic_termination}, which significantly reduces the number of iterations of the algorithm required to achieve an estimate of the capacity with additive error smaller than $\varepsilon$. The results are shown in Figure~\ref{fig:MI_amplitude_damping}. The standard and the adaptive accelerated algorithm converge to the capacity in $12$ and $4$ iterations, respectively.


\section{Conclusion}\label{sec:conclusion}

We give an analytical and numerical analysis of alternating optimization algorithms of Blahut-Arimoto type for computing various channel capacities in quantum information theory. We note that our algorithms are of zeroth-order and do not need to take into account matrix valued derivatives. Thus, they are rather straightforward and computationally inexpensive to implement\,---\,as demonstrated in our numerical examples. It remains open if alternating optimization algorithms, in particular of Blahut-Arimoto type, can also be given for other convex optimization problems in terms of quantum entropy. Starting from the analogous classical settings this could, e.g., be quantum rate distortion functions~\cite{Datta13}, quantum information bottleneck functions~\cite{hirche18}, or certain quantum network capacities~\cite{Winter01}. For example, in the case of the classical information bottleneck, a Blahut-Arimoto type algorithm is used for the estimation of the information bottleneck functions~\cite{Tishby99} and helped to gain insight into the black box of deep neural networks~\cite{tishby_deep_2015,shwartz-ziv_opening_2017}. The quantum information bottleneck was investigated in~\cite{Watanabe16,salek_quantum_2019,hirche18}.

In contrast to classical Shannon theory, in quantum Shannon theory exact capacity formulas are often not known~\cite{wildebook17}. Moreover, known upper and lower bounds are not always in the form of convex optimization problems. For example, the complexity of determining the classical capacity of general entanglement breaking channels is NP-complete~\cite{Beigi07} (see also~\cite{Harrow13} for hardness of approximation results). For such cases, one might rather aim for numerical tools that do well for the average case in practically relevant examples. Alternating optimization algorithms offer an interesting option in this direction. In fact, Nagaoka explored a version of his quantum Blahut-Arimoto algorithm to study the classical capacity of general quantum channels~\cite{Nagaoka01}.\newline


\paragraph*{Acknowledgements.} We thank Joseph Renes and Marco Tomamichel for discussions. RI acknowledges support from the Swiss National Science Foundation through SNSF project No.\ 200020-165843 and through the National Centre of Competence in Research \textit{Quantum Science and Technology} (QSIT).

\paragraph*{Note added.} During finalization of the first version of our work we became aware of the related works~\cite{Cai19, li2019computing} by Li and Cai. 

\appendix

\section{Gibbs' variational principle}
The well-known Gibbs' variational principle is given in the following lemma (which is used to show the form of the optimizer in Lemma~\ref{lem:update_rules}).

\begin{lemma}[Gibbs' principle]\label{lem:Gibbs}
Let $\omega$ a density operator and $H$ a Hermitian matrix on the same space. Then, we have
\begin{align}
\text{$\Tr{\omega H} -\Tr{ \omega \log \omega} \leq \log \Tr{ \expB{H} }$ }
\end{align}
with equality if and only if $\omega=\frac{\expB{H}}{\Tr { \expB{H}}}$.
\end{lemma}


\section{Time complexity} \label{app:time_comp}

The time complexity of Blahut-Arimoto type algorithms is essentially given by the required number of iteration steps times the complexity of one iteration step. In the following, we give a detailed analysis for the computation of the Holevo quantity, the thermodynamic capacity, the coherent information and the mutual information.


\subsection{Holevo quantity} \label{app:time_compl_cq}

The time complexity of Algorithm~\ref{algo:AB_cq_acc} is determined by the required number of iterations and  the time complexity for applying the following update rule
\begin{align} \label{eq:update_rule_cq_complexity}
\lambda^{(t+1)}_i &=  \frac{1}{Z^{(t+1)}}\lambda^{(t)}_i\expB{\frac1\gamma D\left(\tau_{\cE,i} || \cE(\rho_{\lambda^{(t)}})\right)} \\\
&=  \frac{1}{Z^{(t+1)}}\lambda^{(t)}_i\expB{\frac1\gamma \Tr{\tau_{\cE,i} \left\{\log \tau_{\cE,i} - \log\cE(\rho_{\lambda^{(t)}})\right\}}} \, .
\end{align}
Let us sketch how to apply the update rule efficiently to derive an upper bound on the time complexity. We use the following:
\begin{itemize}
\item the channel $\cE_{X \rightarrow B}$ is given as a lookup table, where one can access the elements $\tau_{\cE,i}$ in constant time. The application of $\cE$ to a diagonal density operator $\rho = \sum_{i=1}^{|X|}\rho_{ii} \ket{i}\bra{i}$ can then be calculated as $\cE(\rho)= \sum_{i=1}^{|X|} \rho_{ii} \tau_{\cE,i} $ with time complexity $\mathcal{O}\left(|X||B|^2\right)$,
\item the complexity of calculating the matrix logarithm of $\rho_A$ is $\mathcal{O}(|A|^3)$,
\item the complexity of calculating $\Tr{\rho_A \sigma_A}$ is $\mathcal{O}(|A|^2)$.
\end{itemize}
The algorithm proceeds then with the following calculations:
\begin{enumerate}
\item The terms $c_i =\Tr{\tau_{\cE,i} \log \tau_{\cE,i}}$ can be computed once for all $i \in \{1,2,\dots,|X| \}$ at the start of the algorithm with complexity $\mathcal{O}(|X| |B|^3)$.
\item In each iteration step $t$ we have to compute:
\begin{enumerate}
\item the matrix $H^{(t)}= \log\cE(\rho_{\lambda^{(t)}})$ requiring time $\mathcal{O}(|X||B|^2+|B|^3)$,
\item real numbers $d^{(t)}_i=\Tr{\tau_{\cE,i} \, H^{(t)}}$ for all $i \in \{1,2,\dots,|X| \}$ requiring time $\mathcal{O}(|X||B|^2)$,
\item the unnormalized coefficients $\tilde{\lambda}^{(t+1)}_i = {\lambda}^{(t)}_i\expB{\frac1\gamma (c_i-d_i^{(t)}) }$ requiring time $\mathcal{O}(|X|)$,
\item the normalization coefficient $Z^{(t+1)}= \sum_{i=1}^{|X|} \tilde{\lambda}^{(t+1)}_i$ requiring time $\mathcal{O}(|X|)$.
\end{enumerate}
We conclude that the complexity for one iteration step is $\mathcal{O}(|B|^3+|X||B|^2)$.
\end{enumerate}

The number of required iterations to get an $\varepsilon$-approximation to the capacity is of order $\mathcal{O} ( \log |X|/\varepsilon)$ and hence the required time for all the iteration steps is $\mathcal{O} \left(\left(|B|^3+|B|^2 |X|\right)\log |X|/\varepsilon\right)$. We conclude that the time complexity of Algorithm~\ref{algo:AB_cq_acc} is given by 
\begin{align}
\mathcal{O} \left(\frac{(|B|^3+|B|^2 |X|)\log |X|}{\varepsilon} + |X| |B|^3 \right) \, .
\end{align}


\subsection{Thermodynamic capacity}\label{app:time_compl_thermo}

The time complexity of Algorithm~\ref{algo:AB_thermo_acc} is determined by the required number of iterations and the complexity of applying the following update rule 
\begin{align} \label{eq:update_rule_cq_complexity2}
\rho^{(t+1)} = \frac{1}{Z^{(t+1)}} \expB{\log\rho^{(t)} +\frac1\gamma \HE(\rho^{(t)})} \, ,
\end{align}
where  $\HE(\sigma) = -\log\sigma + \cE^\dagger(\logs \cE(\sigma)) +\log\Gamma_A - \cE^\dagger(\log\Gamma_B)$ and $Z^{(t+1)}=\Tr{\expB{\log\rho^{(t)} +\frac1\gamma \HE(\rho^{(t)})}}$ normalizes the state. To calculate the complexity of applying the update rule, we use that
\begin{itemize}
\item the quantum channel $\cE_{A \rightarrow B}$ and its adjoint $\cE_{B\rightarrow A}^\dagger$ are both given as lookup tables whose $(i,j)$-th entry is given by $\cE(\ket{i}\bra{j}_A)$ or $\cE^\dagger(\ket{i}\bra{j}_B)$, respectively. We assume constant time access to the entries of the table.  Hence, the application of $\cE$ (or $\cE^\dagger$) to a density operator $\rho_A$ has time complexity $\mathcal{O}\left(|A|^2|B|^2\right)$. Indeed, the channel application can be calculated as $\cE(\rho)= \sum_{i,j=1}^{|A|} \rho_{ij} \tau_{\cE,i,j} $ with $\tau_{\cE,i,j}=\cE(\ket{i}\bra{j}_A)$,
\item the complexity of calculating the matrix logarithm and exponential of $\rho_A$ is $\mathcal{O}(|A|^3)$,
\end{itemize}

An iteration step of the form given in~\eqref{eq:update_rule_cq_complexity2} is then found to have complexity $\mathcal{O}(|A|^3 + |A|^2 |B|^2+ |B|^3)$. The number of required iterations to get an $\varepsilon$ approximation to the capacity is of order $\mathcal{O} ( \log |A|/\varepsilon)$ and hence we conclude that the time complexity of Algorithm~\ref{algo:AB_thermo_acc} is given by 
\begin{align}
\mathcal{O} \left(\frac{(|A|^3 + |A|^2 |B|^2+ |B|^3)\log |A|}{\varepsilon}\right) \, .
\end{align}


\subsection{Coherent information}\label{app:time_compl_ci}

The time complexity of Algorithm~\ref{algo:AB_ci_acc} is determined by the required number of iterations and the complexity of applying the following update rule 
\begin{align} \label{eq:update_rule_ci_complexity}
\rho^{(t+1)} = \frac{1}{Z^{(t+1)}} \expB{\log\rho^{(t)} +\frac1\gamma \HE(\rho^{(t)})} \, ,
\end{align}
where  $\HE(\sigma) = \cE_c^\dagger\log\cE_c(\sigma) - \cE^\dagger\log\cE(\sigma)$ and $Z^{(t+1)}=\Tr{\expB{\log\rho^{(t)} +\frac1\gamma \HE(\rho^{(t)}))}}$ normalizes the state. To calculate the complexity of applying the update rule, we use that
\begin{itemize}
\item the application of the quantum channel $\cE_{A \rightarrow B}$ and of its adjoint have time complexity  $\mathcal{O}\left(|A|^2|B|^2\right)$, where the application of the channel $(\cE_c)_{A \rightarrow E}$ and of its adjoint have time complexity $\mathcal{O}\left(|A|^2|E|^2\right)$ (see Appendix~\ref{app:time_compl_thermo} for the argumentation),
\item the complexity of calculating the matrix logarithm and exponential of a density operator $\rho_A$ is $\mathcal{O}(|A|^3)$.
\end{itemize}

An iteration step of the form given in~\eqref{eq:update_rule_ci_complexity} is then found to have complexity $\mathcal{O}(|A|^3 + |B|^3+ |E|^3 +|A|^2 (|B|^2+|E|^2))$. By Stinespring's factorization theorem, the dimension of the environment can be bounded by $|E|\leq K$, where $K\leq |A| |B|$ denotes the Kraus rank of $\cE$, i.e., the minimal number of Kraus operators required to represent $\cE$. Hence, the complexity of one iteration step is  $\mathcal{O}(|A|^3 + |B|^3+ K^3 +|A|^2 (|B|^2+K^2))$.
The number of required iterations to get an $\varepsilon$ approximation to the capacity is of order $\mathcal{O} ( \log |A|/\varepsilon)$ and hence we conclude that the time complexity of Algorithm~\ref{algo:AB_ci_acc} is given by 
\begin{align}
\mathcal{O} \left(\frac{(|A|^3 + |B|^3+ K^3+ |A|^2 (|B|^2+K^2)\log |A|}{\varepsilon}\right) \, .
\end{align}


\subsection{Mutual information}\label{app:time_compl_mi}

The time complexity of Algorithm~\ref{algo:AB_mi_acc} is determined by the required number of iterations and the complexity of applying the following update rule 
\begin{align} \label{eq:update_rule_mi_complexity}
\rho^{(t+1)} = \frac{1}{Z^{(t+1)}} \expB{\log\rho^{(t)} +\frac1\gamma \HE(\rho^{(t)})} \, ,
\end{align}
where  $\HE(\sigma) = -\log(\sigma) + \cE_c^\dagger\log\cE_c(\sigma) - \cE^\dagger\log\cE(\sigma)$ and $Z^{(t+1)}=\Tr{\expB{\log\rho^{(t)} +\frac1\gamma \HE(\rho^{(t)}))}}$ normalizes the state. Note that the complexity of the update rule is the same as in Appendix~\ref{app:time_compl_ci} since the additional term $\log(\sigma)$ in $\HE{(\sigma)}$ does not change the asymptotic complexity. We conclude that the time complexity of Algorithm~\ref{algo:AB_mi_acc} is given by 
\begin{align}
\mathcal{O} \left(\frac{(|A|^3 + |B|^3+ K^3+ |A|^2 (|B|^2+K^2)\log |A|}{\varepsilon}\right) \, .
\end{align}


\bibliographystyle{ultimate2}
\bibliography{new_bib}

\end{document}